\documentclass[sigconf,authorversion,nonacm]{acmart}

\AtBeginDocument{%
  }

\copyrightyear{2024}

\let\sigproof\proof\let\proof\relax
\let\sigendproof\endproof\let\endproof\relax
\usepackage[utf8]{inputenc} 
\usepackage[T1]{fontenc}
\usepackage{url}
\usepackage{ifthen}
\usepackage{amsmath}[cmex10] 
\usepackage{enumitem}

\usepackage{amsfonts,bm,epsfig,graphicx,latexsym,float}
\usepackage{rotating,setspace,latexsym,epsf,color,subcaption}
\usepackage{comment,caption}
\usepackage{soul}

\newcommand{\inner}[2]{\left\langle #1,\, #2 \right\rangle}

\let\proof\sigproof
\let\endproof\sigendproof

\newtheoremstyle{sig}
  {}
  {}
  {\itshape}
  {}
  {\scshape}
  {.}
  {.5em}
  {#1 #2\thmnote{\quad(#3)}}
\newtheorem{remark}{Remark} \newtheorem{Theorem}{Theorem}

\newtheorem{Lemma}{Lemma}

\begin{document}

\title{Age of Information Optimization and State Error Analysis for Correlated Multi-Process Multi-Sensor Systems}

\author{Egemen Erbayat}
\email{erbayat@gwu.edu}
\affiliation{%
  \institution{The George Washington University}
  \streetaddress{}
  \city{Washington}
  \state{DC}
  \country{USA}
  \postcode{}
}

\author{Ali Maatouk}
\email{ali.maatouk@yale.edu}
\affiliation{%
  \institution{Yale University}
  \streetaddress{}
  \city{New Haven}
  \state{Connecticut}
  \country{USA}
  \postcode{}
}
\author{Peng Zou}
\email{pzou94@gwu.edu}
\affiliation{%
  \institution{The George Washington University}
  \streetaddress{}
  \city{Washington}
  \state{DC}
  \country{USA}
  \postcode{}
}
\author{Suresh Subramaniam}
\email{suresh@gwu.edu}
\affiliation{%
  \institution{The George Washington University}
  \streetaddress{}
  \city{Washington}
  \state{DC}
  \country{USA}
  \postcode{}
}

\renewcommand{\shortauthors}{Erbayat et al.}

\begin{abstract}
In this paper, we examine a multi-sensor system where each sensor may monitor more than one time-varying information process and send status updates to a remote monitor over a common channel. We consider that each sensor's status update may contain information about more than one information process in the system subject to the system's constraints. To investigate the impact of this correlation on the overall system's performance, we conduct an analysis of both the average Age of Information (AoI) and source state estimation error at the monitor. Building upon this analysis, we subsequently explore the impact of the packet arrivals, correlation probabilities, and rate of processes' state change on the system's performance. Next, we consider the case where sensors have limited sensing abilities and distribute a portion of their sensing abilities across the different processes. We optimize this distribution to minimize the total AoI of the system. Interestingly, we show that monitoring multiple processes from a single source may not always be beneficial. Our results also reveal that the optimal sensing distribution for diverse arrival rates may exhibit a rapid regime switch, rather than smooth transitions, after crossing critical system values. This highlights the importance of identifying these critical thresholds to ensure effective system performance.
\end{abstract}



\maketitle

\section{Introduction}

In the rapidly developing landscape of networked systems, timeliness plays an essential role in multiple aspects of communication, decision-making, and information processing, contributing significantly to the efficiency and effectiveness of systems. In this realm, the Age of Information (AoI) metric proposed in \cite{yates2012} stands as a pivotal measure, capturing the timeliness of information delivery in communication networks. Due to its importance, the AoI has been well studied in the literature, ranging from single-server systems with single or multiple sources \cite{modiano2015,mm1,sun2016,najm2018,soysal2019,9137714, yates2019,zou2023costly},  to scheduling problems \cite{modiano-sch-1, 9007478,
8845254, sch-igor-1,sch-li,sch-sun} and resource-constrained systems analysis \cite{const-ulukus,const-biyikoglu,const-arafa,const-farazi,const-parisa}.

In sensor networks, collaborative sensing among the different components of the network has been shown to aid in improving the overall performance of the network \cite{collaborativesensing}. Particularly, in such scenarios, numerous small sensor devices are strategically scattered around an area, monitoring different processes and sending updates to one or multiple central controllers \cite{wirelessnetworks}. Home security systems with multiple motion sensors are a good example of how devices can work together to improve efficiency. Each sensor can focus on a specific area and send status updates for that area. However, if there is an overlap between the fields of view of different sensors, they can share information about those areas. This collaboration among devices is referred to as correlation. In scenarios where network resources are constrained, such collaboration can strengthen the system's efficacy and efficiency. It resembles orchestrating a network of compact, intelligent devices working in unison to gather and exchange data, thereby enabling thorough and punctual monitoring.

Given the importance of the AoI in sensor networks, such correlation in status updates has been investigated in the literature.
In \cite{he2018}, a sensor network that has overlapping fields is considered, and the authors presented a joint optimization approach for fog node assignment and transmission scheduling for sensors to minimize the age of multi-view image data. Similarly, in \cite{tong2022}, the authors considered cameras that monitor overlapping areas and propose scheduling algorithms for multi-channel systems for AoI-based minimization. In \cite{popovski2019,modiano2022}, the authors proposed probability-based correlation models and presented sensor scheduling policies aimed at minimizing AoI. In another line of work \cite{ramakanth2023monitoring}, the authors modeled the status updates correlation as a discrete-time Wiener process and proposed a scheduling policy that considers AoI and monitoring error. Despite these contributions, existing studies predominantly assume given correlation parameters and overlook the impact of varying correlation on system performance when addressing scheduling problems. However, optimizing correlation parameters, such as optimizing the placement of the sensors in an area, is also crucial. This underscores the need for further research to investigate how changes in correlation affect system dynamics and to identify optimal correlation parameters under constraints. Our research fills this gap by systematically exploring the effects of correlation variations on system performance. To that end, the main contributions of this paper are summarized as follows:
\begin{itemize}
 \item As a first step, we introduce the system model, taking into account the correlation at hand. Then, by analyzing this system, we present an equivalent process-centric formulation that simplifies the subsequent analysis. 
 \item Following that, we analyze the AoI for each process separately in relation to correlation parameters, formulating closed-form expressions for their averages in the considered M/M/1/1 system. We also consider the estimation error performance for each process separately, given that the AoI is not always a sufficient metric in remote-tracking applications \cite{sun-error}. Our results draw from a stochastic analysis of these metrics that consider all the possible events.
  \item Subsequently, we derive the correlation distribution that minimizes the AoI for three distinct scenarios. Specifically, we investigate three different sensor constraints and determine the optimal solution for each case.
   \item Furthermore, we present numerical implementations to validate the closed-form expressions we derived. We compare different parameter configurations, focusing on their average AoI and error ratio. Our results highlight the impact of status updates correlation on both the AoI and estimation error.
   \item  Lastly, we investigate the optimal correlation distribution and implement the derived optimal policies. Interestingly, our implementations showcase that the optimal distribution policy undergoes a significant regime shift beyond a specific parameters threshold. This observation highlights the need for adaptive, contextually-sensitive strategies in navigating optimal solution spaces. 
\end{itemize}The rest of the paper is organized as follows. We present the system model in Section \ref{system-model}. Afterward, we formulate the equivalent and simplified system in Section \ref{reduction}. The analysis of AoI and error ratio is then conducted in Section \ref{aoi-S} and Section \ref{error-s}, respectively. In Section \ref{opt-s}, we put our optimization problem into perspective and propose solutions to find the optimal sensing distribution. Finally, we present the numerical results in Section \ref{numerical-s} while Section \ref{conc-s} concludes the paper.

\begin{figure}[!t]
  \centering
  \includegraphics[width=0.37\textwidth]{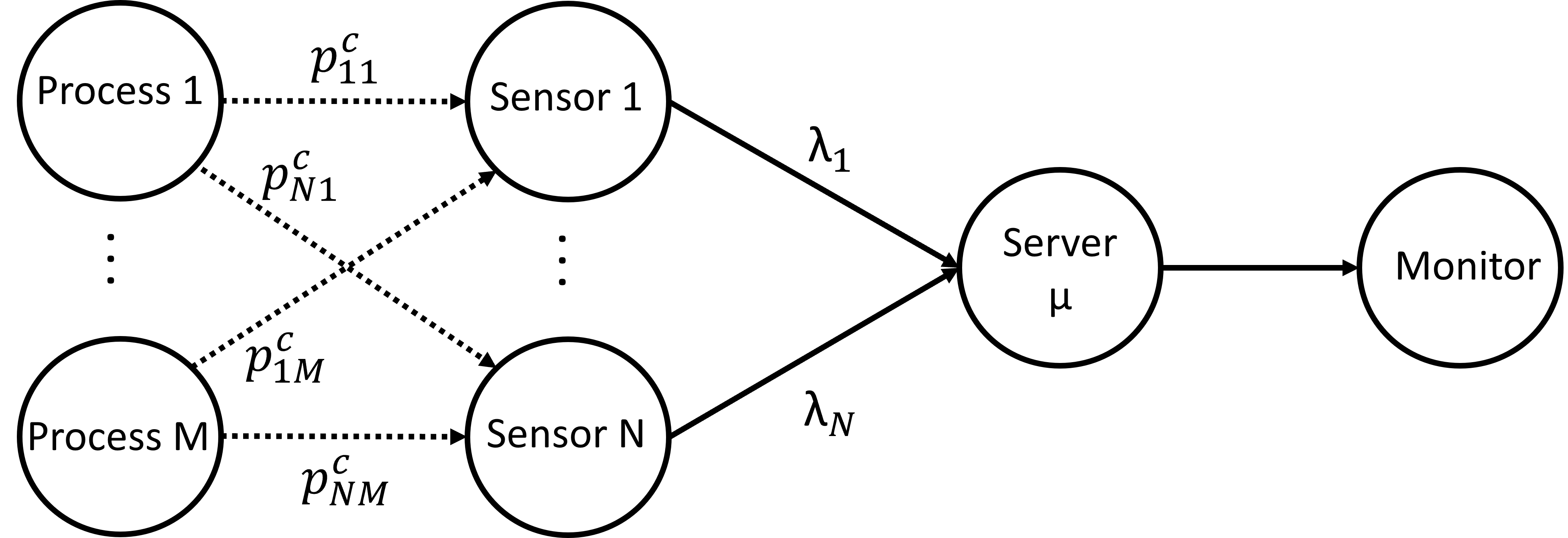}
  \caption{Illustration of our system model.}
  \label{fig:systemmodel}
\end{figure}
\section{System Model} \label{system-model}

Let us consider a sensor network where $N$ sensors track $M$ information processes. To keep the monitor updated, each sensor generates status updates and sends them through a common server/channel, as shown in Figure \ref{fig:systemmodel}. 
We consider that the service time of each packet is exponentially distributed with a service rate $\mu$. We also assume that sensor $i$ generates packets according to a Poisson process of rate $\lambda_i$. We adopt a zero-buffer assumption for the server in our model. This choice is motivated by previous research demonstrating its optimality for AoI minimization in certain scenarios, such as single information source systems with preemption \cite{bedewy}. While this optimality does not extend to our model, our initial numerical investigations have revealed that incorporating a buffer does not consistently enhance performance, as detailed in Appendix A. Consequently, we maintain the zero-buffer assumption throughout our analysis. Accordingly, any arriving packet that finds the server busy is dropped \cite{dataNetworks:book}. With all the above in mind, we define \(\boldsymbol{\lambda}\) as a vector representing the arrival rates from the different sensors, where \(\lambda_i\) is the arrival rate from sensor \(i\) for $i=1,
\ldots,N$. Specifically, we have:
\begin{equation}
\boldsymbol{\lambda}^T = \begin{bmatrix}
\lambda_{1} \quad \lambda_{2} \quad \dots \quad \lambda_{N} 
\end{bmatrix}.
\end{equation}

As for the information process, we consider that each physical process evolves as a time-varying discrete stochastic process. Particularly, the physical process $j$ is modeled as a Markov chain with $K$ different states. To represent these state changes, we use $\Omega^j_{ab}$ to denote the transition probability from state $a$ to state $b$ of process $j$.  In matrix form, the transition matrix \(\boldsymbol{\Omega_j} \in [0,1]^{K\times K}\) can be defined as follows: 
\begin{align}
\boldsymbol{\Omega_j} = {\small \begin{bmatrix}
\Omega^j_{11} & \Omega^j_{12} & \dots & \Omega^j_{1K} \\
\Omega^j_{21} & \Omega^j_{22} & \dots & \Omega^j_{2K} \\
\vdots & \vdots & \ddots & \vdots \\
\Omega^j_{K1} & \Omega^j_{K2} & \dots & \Omega^j_{KK}
\end{bmatrix}} \normalsize, \quad \text{for } j=1,\ldots, M.
\end{align}
In this paper, we consider that the Markov chain is irreducible and aperiodic. To that end, we can conclude the existence and uniqueness of the chain's stationary distribution. We denote the stationary distribution of the Markov chain formed with $\boldsymbol{\Omega_j}$ by:
\begin{equation}
\boldsymbol{\psi_j} = \begin{bmatrix}
\psi^j_{1} &\psi^j_{2} & \dots & \psi^j_{K}
\end{bmatrix}
\label{stationary_state}.
\end{equation}
This stationary distribution can obtained by solving the equation:
\begin{equation}
\boldsymbol{\psi_j}\cdot\boldsymbol{\Omega_j}   = \boldsymbol{\psi_j},
\end{equation}
and normalizing the resulting vector to ensure that the sum of its components is equal to 1 \cite{eigenvalue}.


We assume that state transitions for process $j$ occur after exponentially distributed epochs with a rate of $\zeta_{j}$. Accordingly, the generation of status updates by sensors and information process changes are decoupled, reflecting scenarios where sensors observe multiple processes simultaneously, such as a camera tracking various activities. This decoupling allows each process's state to evolve independently, regardless of active tracking, ensuring a more resilient and adaptable system. It accommodates situations where sensors may not detect every change or update for each detected change, and where processes evolve at different rates. Ultimately, this decoupling creates a robust, realistic model that better adapts to the complexities and limitations of multi-process monitoring systems, accurately representing practical constraints in real-world sensing applications.


Finally, to model the correlation among the different sensor observations, we assume that each packet generated by sensor $i$ contains information about the process $j$ with a correlation probability $p^c_{ij}$. The information the packet has is the state of the processes at the generation time of the packet. 
To that end, we define the correlation matrix \(\mathbf{P_C} \in [0,1]^{N\times M}\) as follows:
\begin{equation} 
\mathbf{P_C} = {\small\begin{bmatrix}
p^c_{11} & p^c_{12} & \dots & p^c_{1M} \\
p^c_{21} & p^c_{22} & \dots & p^c_{2M} \\
\vdots & \vdots & \ddots & \vdots \\
 p^c_{N1} & p^c_{N2} & \dots & p^c_{NM}
\end{bmatrix}.}
\end{equation}
Key symbols defined in this paper are summarized in Table \ref{tab:key_symbols}. After having outlined the system model, we now proceed to formulate the equivalent and simplified system in Section \ref{reduction}.

\begin{table}[h]
    \centering
    \caption{Key Symbols}
    \label{tab:key_symbols}
    \small 
    \begin{tabular}{|c|l|}
        \hline
        \textbf{Symbol} & \textbf{Definition} \\
        \hline
        $\lambda_i$ & Arrival rate of sensor $i$ \\
        $\boldsymbol{\Omega_j}$ & State transition probabilities of process $j$ \\
        $\zeta_j$ & State change rate of process $j$ \\
        $\boldsymbol{\psi_j}$ & Stationary distribution of $\boldsymbol{\Omega_j}$ \\
        $\mathbf{P_C}$ & Correlation probabilities among sensors and processes \\
        $p^c_{ij}$ & Correlation probability between sensor $i$ and process $j$ \\
        $\lambda_C$ & Sum of sensor arrival rates \\
        $\lambda^*_j$ & Rate of informative arrivals for process $j$ \\
        $\Tilde{p_j}$ & Probability of serving informative packet for process $j$ \\
        \hline
    \end{tabular}
    \normalsize 
\end{table}
\section{System Simplification through Equivalence } \label{reduction}

In the considered system, the originator of the packet containing information about any arbitrary process $j$ is irrelevant from the monitor's perspective. In fact, concerning process $j$, what matters to the monitor is whether the served packet contains information about process $j$ or not rather than which sensor provided the update. To that end, we label a status update as informative for process $j$ if it contains information on the process $j$. Otherwise, we label it as uninformative.  Building on this concept, we define the informative arrival rate vector $\boldsymbol{\lambda^*}$ as follows:
\begin{equation}
{\boldsymbol{\lambda^{*}}}^T = \begin{bmatrix}
\lambda_{1}^* & \lambda_{2}^* & \dots & \lambda_{M}^*\end{bmatrix} = \boldsymbol{\lambda}^T\mathbf{P_C},
\end{equation}
where $\lambda^*_j$ represents the informative arrival rate for process $j$. As a last step, we let $\lambda_C$ be the arrival rate of the server, which is the sum of all arrival rates, as follows:
\begin{equation}
\lambda_C = \sum_{i=1}^{N} \lambda_i.
\end{equation}
With the above entities in mind, we provide the following system equivalence lemma. 

\begin{Lemma}\label{Lem1}
$\frac{}{}$
Consider a process $j$ among $M$ processes. From the monitor's perspective, the system is equivalent to Figure \ref{fig:equv_model}, where there are two packet sources: 
\begin{itemize}
    \item packets with information, with a rate of $\lambda^*_j$, and 
    \item packets without information with a rate of  $\lambda_C - \lambda^*_j$.
\end{itemize}
\end{Lemma}
\begin{proof} The details can be found in Appendix B.
\end{proof} 
Using the above equivalence, we can analyze the environment by reducing the original system to $M$ independent systems, each with two sources as depicted in Fig. \ref{fig:equv_model}.
The independence of these 
M systems stems from the Poisson nature of packet arrivals. From the perspective of any single process $j$, the arrivals of both informative and uninformative packets from all other processes can be shown to collectively form Poisson streams, as detailed in Appendix B. 
With this in mind, in the next section, we derive a closed-form expression of the AoI for each process, taking into account both informative and uninformative status updates to comprehensively evaluate their impact on the system's dynamics.
\begin{figure}[!t]
  \centering
  \includegraphics[width=0.37\textwidth]{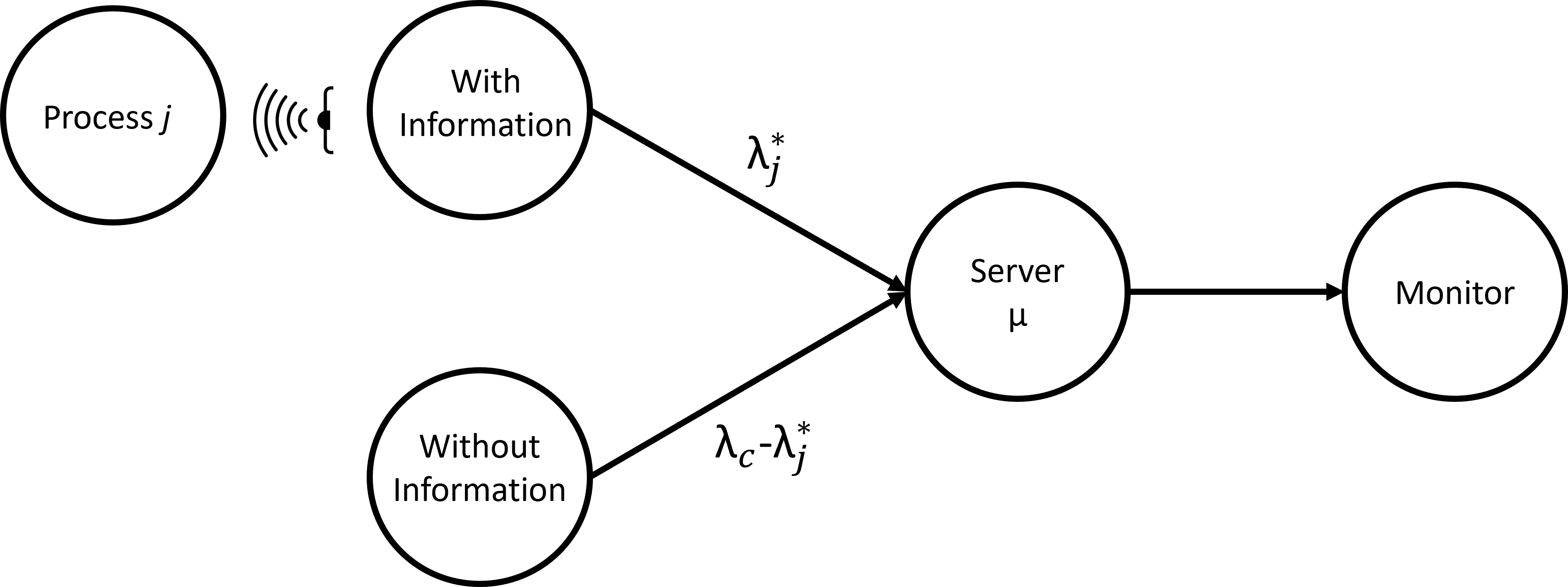}
  \caption{Equivalent system model from process $j$'s perspective.
  }
  \label{fig:equv_model}
\end{figure}
\section{Age of Information Analysis}\label{aoi-S}

In this section, we consider the age function introduced in \cite{yates2012} as a performance metric. Mathematically, the AoI of process $j$ at time $t$, denoted by $\Delta_j(t)$, can be defined as:
\begin{equation}
\Delta_j(t) = t - T_j,
\end{equation}
where $T_j$ represents the time at which the most recent informative packet for process $j$ was generated. Particularly, the age at the monitor for each process $j$ increases linearly over time until an informative status update is received, upon which a drop in the age takes place. As mentioned in Section \ref{system-model} and Section \ref{reduction}, the packet in the server may or may not have information about each process. If the served packet has information about process $j$, the AoI for process $j$ decreases just after the end of the service time. However, if the served packet has no information on process $j$, the AoI for process $j$ continues to increase linearly.
 Let \(t_k\) denote the time instant when the \(k\)-th packet is generated, and \(t^\prime_k\) represent the time instant when this packet completes service. When the server is busy with an informative or uninformative packet upon the arrival of a new packet, the new packet is dropped. To that end, we denote by \(t^d_n\) the time instant when the \(n\)-th dropped packet was generated. 
We define \(T_k\) as the service time of the $k$-th packet, given by
\begin{equation}
T_k := t^\prime_k - t_k, 
\end{equation}
and \(Y_k\) as the inter-departure time between two consecutive packets, given by
\begin{equation}
Y_k := t^\prime_{k} - t^\prime_{k-1}.
\end{equation}
We also define $\tilde{Y}_{j}^{l}$ as the inter-departure time between the $l$-th informative packet and the $(l-1)-$th informative packet from process $j$'s perspective. 
Since \(Y_k\) shares the same distribution for all $k$, we define the random variable \(Y\) to represent them collectively. Similarly, considering that \(\Tilde{Y}_{j}^l\) shares the same distribution for all $l$, we define the random variable \(\Tilde{Y}_{j}\) to represent them as a group. 
To understand the AoI process better, we illustrate the evolution of the AoI in Figure \ref{fig:sample_path}. The age of information for process $1$ at the destination node follows a linear increase over time. When a new informative status update is received, the age is reset to the time difference between the current time instant and the timestamp of the received update ($A_1$). However, if the status update is uninformative for process $1$ like the update completed at $t^\prime_2$, the age continues to increase linearly ($A_2$). The packets arriving at times $t^d_1$ and $t^d_2$ are dropped. The server is occupied with an uninformative packet at time $t^d_1$ and an informative packet at time $t^d_2$ for process $1$. In addition, $\tilde{Y}_{1}^{2}$ is the second informative interarrival time which is equal to the time difference between $t^\prime_3$ and $t^\prime_1$ that are the second and the first informative departures.

Next, we define the effective arrival rate as the rate of packets that arrive when the server is idle.
Let \( \lambda^e_{j}\) be the effective arrival rate for packets that are informative for process $j$. Consequently, we have 
\begin{equation}
\lambda^e_{j} = \frac{\mu\lambda^*_j}{\mu + \lambda_C}, \quad \textnormal{for }j=1,\ldots,M.
\end{equation}

Given the above quantity, we derive below the average AoI for each process $j$. 
\begin{Lemma}\label{Lem2}
The average AoI $\Delta_j$ for process $j$ is:
\begin{align}
 \Delta_j = \lambda^e_{j}\left(\frac{1}{2}\mathbb{E}[\Tilde{Y}_j^2] + \frac{\mathbb{E}[\Tilde{Y}_j]}{\mu}\right), \quad \textnormal{for }j=1,\ldots,M.
     \label{ageofinformation}
\end{align}
\end{Lemma}

\begin{proof}

We adapt the methodology outlined in \cite{mm1} to our scenario by leveraging their approach of utilizing inter-arrival times for calculating the age of information. Despite the absence of informative categorization in \cite{mm1}, their method of deriving the age of information based on inter-arrival times remains applicable. The main modification involves substituting variables to align with our informative arrivals.
\end{proof}

\begin{figure}[t]
  \centering
  \includegraphics[width=0.35\textwidth]{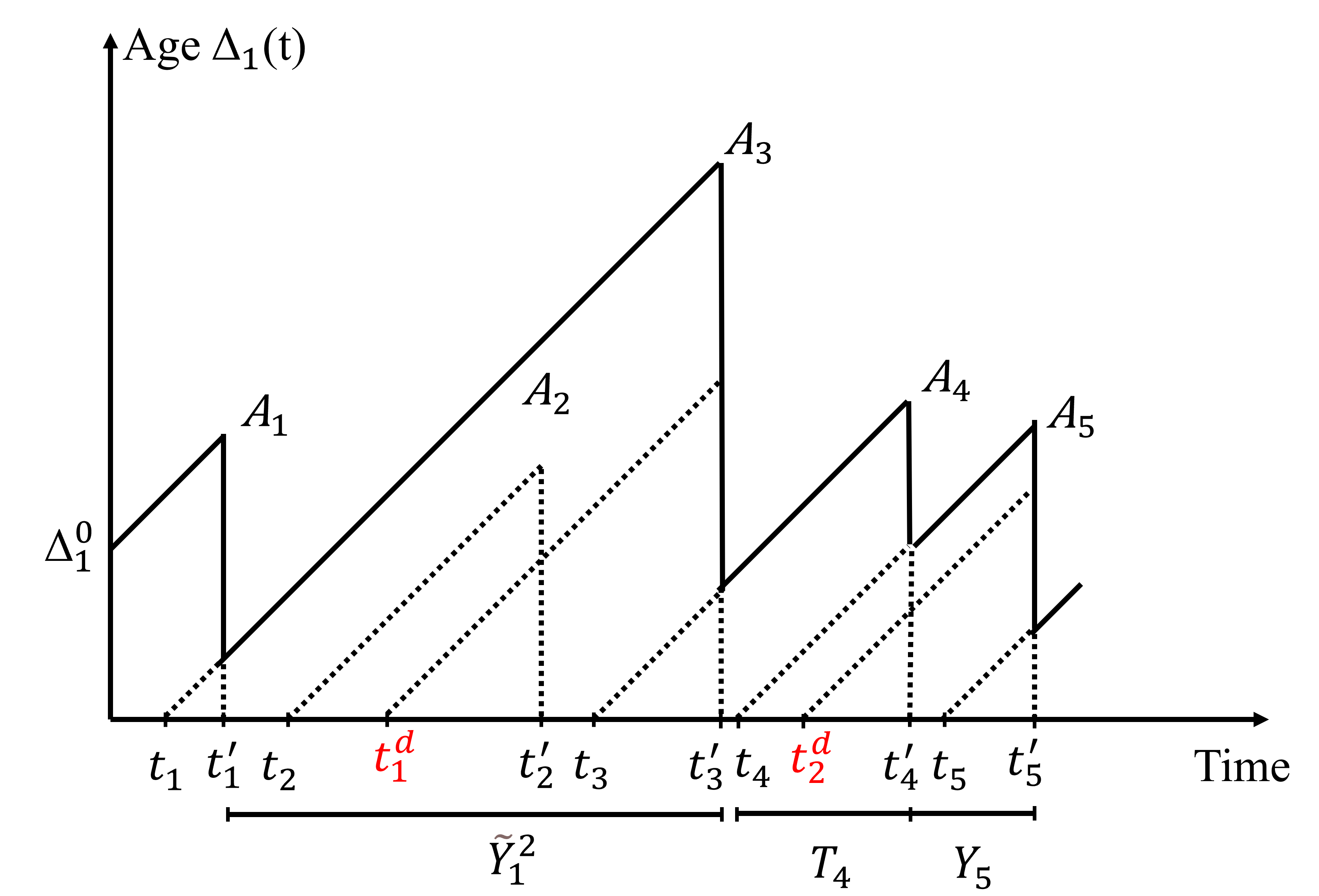}
  \caption{Evolution of AoI for process 1.}
  \label{fig:sample_path}
\end{figure}
Next, we define $\Tilde{p_j}$ as the probability of serving an informative packet for process $j$, as shown below:
\begin{equation}
\Tilde{p_j} = \frac{\lambda^*_j}{\lambda_C}.
\end{equation}
After defining $\Tilde{p_j}$, we utilize it to establish the relationship between $\Tilde{Y}_j$ and $Y$. The details of this relationship can be found in  Appendix C. Using this relationship, we present the AoI for process $j$ in Theorem
\ref{The1}, formulated in terms of $\mu$, $Y_j$,  and $\lambda^e_{j}$.
 \begin{Theorem}\label{The1}
In the considered M/M/1/1 system, the average AoI for process $j$ denoted as $\Delta_j$ is: 
  \begin{align}
 \Delta_j = \frac{\lambda_C}{\lambda_C+\mu}\left(\frac{\mu\mathbb{E}[Y^2]}{2} + \frac{\mu\mathbb{E}[Y]^2(1-\Tilde{p_j})}{\Tilde{p_j}} + \mathbb{E}[Y]\right), \quad \textnormal{for }j=1,\ldots,M. \label{final-aoi-eq}
\end{align}

\end{Theorem}
	\begin{proof} The details can be found in Appendix C.
\end{proof}

 To provide an interpretation of the above formula, we can see that if the value of $\Tilde{p_j}$ is 1, it means that every packet is informative for process $j$. In such a scenario, the system acts like a single sensor system, continuously sending status updates of process $j$ from its perspective.  On the other hand, as $\Tilde{p_j}$ approaches 0, the packets become uninformative. As a result, the AoI approaches infinity, since informative status updates are received more and more infrequently.

\section{Error Ratio Analysis}\label{error-s}

We define a binary function $\epsilon_j(t)$ such that if the state information the monitor has for process $j$ is the same as the state of process $j$ at time $t$, then $\epsilon_j(t)$ is equal to 1. Otherwise, it is equal to $0$. Then, we define the error of process $j$ as the ratio of the total duration when $\epsilon_j(t)$ is equal $0$ over the entire time horizon, denoted by $\epsilon_j$. Particularly, we have
\begin{equation}
\epsilon_j = 1 - \lim_{{T \to \infty}} \frac{1}{T} \int_{{0}}^{{T}} \epsilon_j(t) \, dt , \quad \textnormal{for }j=1,\ldots,M.
\end{equation}
Due to the nature of our system, \(\epsilon_i\) and \(\epsilon_j\) are independent for any \(i, j \in \{1, \ldots, M\} \) with \(i \neq j\). The reason behind that is that state changes of two processes are independent of each other, and the system functions as $M$ different independent systems as demonstrated in Section \ref{reduction}. Therefore, we derive the generic error \(\epsilon\) for any process to simplify our analysis. Particularly, we drop the index $j$ of the considered entities, and we use $\epsilon$, $\epsilon(t)$, $\zeta$, $\lambda^*$ and $\boldsymbol{\Omega}$ to denote the system parameters in the remainder of this section.

To find $\epsilon$ analytically, we investigate a Markov Chain that considers the current process state, the state at the monitor, and the state of the served packet (defined below). Our proposed Markov chain is 3-dimensional with dimensions $(x,y,z)$. The states \(x\) and \(y\) represent the current state and the monitor state, respectively, and each can take values between 1 and \(K\). The state \(z\) can have three possible values as depicted below: 
\begin{itemize}
    \item $z=0$: This state indicates that the server is currently idle.
    \item $z=1$: This state signifies that the server is actively serving a packet containing information from the process of interest.
    \item $z=2$: In this state, the server is occupied by a packet, but this packet does not contain information about the process in question. In other words, it carries information about other processes.
\end{itemize}
Note that when $x$ and $y$ are equal, the current state and the state the monitor has are identical. To this end, we redefine $\epsilon(t)$ as 
\begin{equation}
\epsilon(t) = \begin{cases}
    1 & \text{if } x = y, \\
    0 & \text{otherwise}.
\end{cases}
\end{equation}
With all the above in mind, we note that the system has three types of events: packet arrivals, packet departures, and state changes. Each event causes a transition in this three-dimensional Markov Chain. Let $\mathbf{P_M}$ and $\pi$ be the transition probability matrix corresponding to those transitions and the stationary distribution of the three-dimensional Markov chain. 
Given that the one-dimensional Markov chain shown in Section \ref{system-model} characterized by both irreducibility and aperiodicity is used to form the considered three-dimensional Markov Chain, the three-dimensional Markov chain can be shown to be irreducible and aperiodic. Next, we let $\pi(x, y, z)$ denote the stationary probability of being at state $(x,y,z)$. Then, by definition, the following equation is verified:
\begin{equation}
\sum_{x=1}^{K} \sum_{y=1}^{K} \sum_{z=0}^{2} \pi(x, y, z) = 1. 
\end{equation}
To derive this stationary distribution, we first need to calculate the probability of state change until the packet is served, given that the server is occupied with an informative packet. To do this, we need to find the probability of the process transitioning from state $i$ to state $j$, denoted as $p^n_{ij}$, while an informative packet is being served. We derive $p^n_{ij} $ for all $\quad i = 1, \dots, N$ and $j = 1, \dots, M$ in the following lemma, Lemma \ref{Lem3}. Afterward, and by leveraging these results, we formulate $\mathbf{P_M}$. 

\begin{Lemma}\label{Lem3}
Let \(\mathbf{P_N} \in [0,1]^{K\times K}\) represent the matrix of elements $p^n_{ij}$'s. The matrix can be obtained as follows: 
\begin{equation}
\mathbf{P_N} =
{\small
\begin{bmatrix}
p^n_{11} & p^n_{12} & \dots & p^n_{1K} \\
p^n_{21} & p^n_{22} & \dots & p^n_{2K} \\
\vdots & \vdots & \ddots & \vdots \\
p^n_{K1} & p^n_{K2} & \dots & p^n_{KK}
\end{bmatrix}} = \frac{\mu}{\mu+\zeta}  \left(\mathbf{I} - \frac{\zeta \boldsymbol{\Omega}}{\mu+\zeta}\right)^{-1}.
\end{equation}
\end{Lemma}
\begin{proof} The details can be found in Appendix D.
\end{proof}

\begin{Lemma}\label{Lem4}
Let $\mathbf{P_M}_{(x_1,y_1,z_1)\rightarrow(x_2,y_2,z_2)} \in [0,1]$ be the transition probability from $(x_1,y_1,z_1)$ to $(x_2,y_2,z_2)$. $\mathbf{P_M}$ can be obtained as follows:\\
For every \( x_1,x_2, y_1 = 1, \ldots, K \):
\begin{equation}
\mathbf{P_M}_{(x_1,y_1,0)\rightarrow(x_2,y_1,0)} = \Omega_{x_1x_2} \frac{\zeta}{\zeta + \lambda_C},
\label{0_state_eq}
\end{equation}

\begin{equation}
\mathbf{P_M}_{(x_1,y_1,1)\rightarrow(x_2,y_1,1)} = \Omega_{x_1x_2}\frac{\zeta}{\zeta + \mu},
\label{1_state_eq}
\end{equation}
\begin{equation}
\mathbf{P_M}_{(x_1,y_1,2)\rightarrow(x_2,y_1,2)} = \Omega_{x_1x_2} \frac{\zeta}{\zeta + \mu}.
\label{2_state_eq}
\end{equation}
For every \( x_1, y_1 = 1, \ldots, K \):
\begin{equation}
\mathbf{P_M}_{(x_1,y_1,0)\rightarrow(x_1,y_1,1)} = \frac{\lambda^*}{\zeta + \lambda_C},
\label{informative_arrival_eq}
\end{equation}
\begin{equation}
\mathbf{P_M}_{(x_1,y_1,0)\rightarrow(x_1,y_1,2)} = \frac{\lambda_C-\lambda^*}{\zeta + \lambda_C},
\label{uninformative_arrival_eq}
\end{equation}
\begin{equation}
\mathbf{P_M}_{(x_1,y_1,2)\rightarrow(x_1,y_1,0)} = \frac{\mu}{\zeta + \mu}.
\label{uninformative_departure_eq}
\end{equation}
For every \( x_1, y_1,y_2 = 1, \ldots, K \):
\begin{equation}
\mathbf{P_M}_{(x_1,y_1,1)\rightarrow(x_1,y_2,0)} = \frac{\mu}{\zeta + \mu}p^n_{y_2x_1}\frac{\psi_{y_2}}{\psi_{x_1}}.
\label{informative_departure_eq}
\end{equation}
Otherwise:
\begin{equation}
\mathbf{P_M}_{(x_1,y_1,z_1)\rightarrow(x_2,y_2,z_2)} = 0.
\label{other_Events_eq}
\end{equation}
\end{Lemma}
\begin{proof} The details can be found in Appendix E.
\end{proof}
The distinction between matrices $\mathbf{P_M}$ and $\mathbf{P_N}$ is crucial here. $\mathbf{P_M}$ represents the transition probabilities in the three-dimensional Markov chain, considering the server state (idle, serving the process of interest, or serving other processes). On the other hand, $\mathbf{P_N}$ is a matrix that captures the state transition probabilities of the process itself during the time an informative packet is being served.

With $\mathbf{P_M}$ formulated, we can now obtain the stationary distribution $\pi(x, y, z)$. Note that we cannot directly use $\pi(x, y, z)$ to calculate the error over the entire time span because of the embedded nature of the Markov chain \cite{embeddedmarkov}. That is, the state transitions in $\mathbf{P_M}$ do not occur at every small time step; instead, they occur only with events of packet arrival, packet departure, and process state change, which makes the time spent in each $(x,y,z)$ state at a single visit different from each other. Therefore, we need to incorporate the length of time spent in each $(x,y,z)$ state. To that end, we let $w{(x,y,z)}$ denote the weighted holding time at state $(x,y,z)$ over the entire time span. In other words, ${W \times w{(x,y,z)}}$ represents the expected time waited at state $(x,y,z)$ during $W$ state transitions. Note that all these events, i.e., packet arrival, packet departure, and state change of the process, are memoryless, independent, and identically distributed, which preserves Markov property in the embedded chain. Due to the memoryless property, only the occurrence time of the next event has an impact on the time spent in each state. 
With all the above in mind, we define a weighted holding time function $w{(x,y,z)} = \pi(x, y, z)\mathbb{E}[T_{(x, y, z)}]$ where $\mathbb{E}[T_{(x, y, z)}]$ is the expected holding time that represents the expected time waited until the next jump in state $(x,y,z)$ at every transition. It is important to note that the holding time in state $(x,y,z)$ is independent of $\pi(x, y, z)$. We present expected holding times in Lemma \ref{Lem5}. $\mathbb{E}[T_{(x, y, z)}]$ values vary with different $z$ values, but changes in $x$ or $y$ do not affect $\mathbb{E}[T_{(x, y, z)}]$ because the possible events that can happen are the same for the same $z$. For example, the possible events are packet arrivals and state changes of the process if $z=0$ for all $x,y = 1, \dots K$.
\begin{Lemma}\label{Lem5}
Let $\mathbb{E}[T_{(x,y,z)}]$ be the expected holding time in state $(x,y,z)$. Then,  $\mathbb{E}[T_{(x,y,z)}]$ can be calculated for every \( x,y = 1, \ldots, K \) as follows:
\begin{equation}
\mathbb{E}[T_{(x,y, 0)}] = \frac{1}{\zeta+\lambda_C},    
\end{equation}
\begin{equation}
\mathbb{E}[T_{(x,y, 1)}] = \mathbb{E}[T_{(x,y, 2)}]  = \frac{1}{\zeta+\mu}. 
\end{equation}

\end{Lemma}
\begin{proof} The details can be found in Appendix F.
\end{proof}

Then, the error ratio $\epsilon$ becomes the ratio of the sum of weighted holding times when $x$ and $y$ are not equal to the sum of all weighted holding times as follows:
\begin{equation}
\epsilon = \frac{\sum_{\substack{y \neq x}} \sum_{z=0}^{2} w{(x,y,z)}}{ \sum_{x=1}^{K} \sum_{y=1}^{K} \sum_{z=0}^{2} w{(x,y,z)}}.
\end{equation}

\section{Average Age Optimization}\label{opt-s}
Correlation is an important factor in minimizing AoI, as demonstrated in equation (\ref{final-aoi-eq}) and decisions such as the placement of sensors can affect the correlation among them. It is obvious that a higher correlation leads to a smaller AoI. However, increasing the correlation between two processes in a sensor's status updates could decrease the correlation between other processes due to constraints on the sensor's sensing capabilities. Therefore, it is crucial to distribute the correlation probabilities adequately to minimize AoI. In this section, we explore how to assign $\mathbf{P_C}$ under certain constraints representing different scenarios to minimize the sum average AoI. Let $\Delta_{sum}$ denote the sum AoI, depicted below:
  \begin{align}
 \Delta_{sum} = \sum_{j = 1}^{M} \frac{\lambda_C}{\lambda_C+\mu}\big(\frac{\mu\mathbb{E}[Y^2]}{2}-\mu\mathbb{E}[Y] + \mathbb{E}[Y]+ \frac{\mu\mathbb{E}[Y]}{\Tilde{p}_j}\big).
\end{align}
To that end, the objective is to solve the following problem: \begin{align*}
\min_{\mathbf{P_C} \in [0,1]^{N\times M}} \Delta_{sum}, \\  \text{s.t. 
 } h_i(\mathbf{P_C}) \leq 0 \frac{}{},\forall_{i \in [N]} \nonumber, \end{align*} 
 where $h_i(\mathbf{P_C})$ represents the $i$-th sensor's sensing ability constraint. Given that each sensor has its own constraint, the optimal solution has to meet all $N$ sensor constraints. We keep $h_i(\mathbf{P_C})$ general for now. However, later in this section, we define three different $h_i(\mathbf{P_C})$ functions for representing different scenarios.
 
Given that the only parameters affected by variable $\mathbf{P_C}$ in $\Delta_{sum}$ are $\Tilde{p}_j$'s, we can remove the other parts to simplify the problem. Specifically, the optimal $\mathbf{P_C}$ values are the same resulting from solving the problem below:
\begin{align}
\min_{\mathbf{P_C} \in \mathbb{R}^{N\times M}} f(\mathbf{P_C}) =\sum_{j = 1}^{M}\frac{1}{\Tilde{p}_j}, \label{minimization}\\  \text{s.t. } h_i(\mathbf{P_C}) \leq 0 \frac{}{},\forall_{i \in [N]}, \nonumber 
\end{align}
\begin{align}
0 \leq p^c_{ij} \leq 1, \nonumber \quad i = 1, \dots, N,\quad j = 1, \dots, M, \end{align} 
where $\begin{bmatrix}
\Tilde{p}_1 & \Tilde{p}_2 & \dots & \Tilde{p}_M\end{bmatrix} = \frac{\boldsymbol{\lambda}^T\mathbf{P_C}}{\lambda_C}$ and the constraint $ h_i(\mathbf{P_C}) \leq 0$ represents the constraint for the sensor $i$. To pursue our analysis, we start by proving the convexity of the ojective function $f$ in the following lemma. 
\begin{Lemma}\label{Lem6}
The objective function $f$ is a convex function. 
\end{Lemma}
\begin{proof} The details can be found in Appendix G.
\end{proof} 
Even if our objective function is convex, the convexity of the problem relies on whether or not the sensor constraints preserve convexity. 
When dealing with convex problems, the Karush–Kuhn–Tucker (KKT) conditions are sufficient for optimality \cite{boyd2004convex}. To that end, we derive the KKT conditions and we formulate the Lagrange function of the optimization problem as follows:
\begin{align}
    \mathcal{L}(\mathbf{P_C},\boldsymbol{\tau},\mathbf{v},\boldsymbol{\xi}) =\sum_{j = 1}^{M}\frac{\lambda_C}{\sum_{i = 1}^{N} p^c_{ij}\lambda_i} + \sum_{i = 1}^{N}\sum_{j = 1}^{M} (p^c_{ij} - 1 )\tau_{ij} \nonumber \\
 - \sum_{i = 1}^{N}\sum_{j = 1}^{M} p^c_{ij}v_{ij} + \sum_{i = 1}^{N}h_i(\mathbf{P_C})\xi_i.
\end{align}
The KKT conditions for the optimization problem described in eq. (\ref{minimization}) are as follows. $^*$ is used to show optimum variables. 
\begin{align}
 \tau^*_{ij} - v^*_{ij} - \frac{\lambda_C\lambda_i}{(\sum_{k = 1}^{N} p^{c*}_{kj}\lambda_k)^2} + \sum_{k = 1}^{N}\xi^*_k \frac{d}{dp^c_{ij}} h_k(\mathbf{P_C^*}) = 0, \nonumber\\ \quad i = 1, \dots, N, \quad j = 1, \dots, M, \label{kkt-derivative1}
\end{align}
\begin{align}
 (p^{c*}_{ij} - 1 )\tau^*_{ij} = 0,\quad i = 1, \dots, N, \quad j = 1, \dots, M, \label{kkt-cs1}
\end{align}
\begin{align}
 p^{c*}_{ij}v^*_{ij} = 0,\quad i = 1, \dots, N, \quad j = 1, \dots, M, \label{kkt-cs2}
\end{align}
\begin{align}
 h_i(\mathbf{P_C^*})\xi^*_i = 0, \quad i = 1, \dots, N, \label{kkt-cs3}
\end{align}
\begin{align}
    \boldsymbol{\tau}^*,\mathbf{v}^*,\boldsymbol{\xi}^* \geq 0, 
\end{align}
\begin{align} 1 \geq \mathbf{P_C^*} \geq 0, \end{align}
\begin{align} h_i(\mathbf{P_C^*}) \leq 0 \quad i = 1, \dots, N. \end{align} 

In the sequel, we proceed under the assumption that the sensing ability constraint of each sensor operates independently of each other. This assumption is made to establish scenarios where each sensor's sensing ability is solely dependent on itself, without influence from other devices. Notably, $row_i(\mathbf{P_C})$ denotes the probability assignments of sensor $i$ across all processes. Consequently, $h_i(\mathbf{P_C})$ exclusively comprises variables from $row_i(\mathbf{P_C})$ when considering independent sensors. Hence, we can reformulate equation (\ref{kkt-derivative1}) for independent sensors as follows:
\begin{align}
 \tau^*_{ij} - v^*_{ij} -\frac{\lambda_C\lambda_j}{(\sum_{k = 1}^{N} p^{c*}_{kj}\lambda_k)^2} + \xi^*_i \frac{d}{dp^c_{ij}} h_i(\mathbf{P_C^*}) = 0, \nonumber\\ \quad i = 1, \dots, N, \quad j = 1, \dots, M. \label{KKT-derivative2}
\end{align}
Let us now analyze these conditions. To start, consider the scenario where $p^{c*}_{ij} \neq 1$, leading to the outcome $\tau^*_{ij} = 0$. Conversely, if $p^{c*}_{ij} \neq 0$, it results in $v^*_{ij} = 0$ in eq. (\ref{kkt-cs1}) and (\ref{kkt-cs2}). This dichotomy highlights that at least one of $\tau^*_{ij}$ and $v^*_{ij}$ must be zero, given that $p^{c*}_{ij}$ cannot simultaneously be $0$ and $1$. In equation (\ref{KKT-derivative2}), the non-zero nature of $\frac{\lambda_C\lambda_j}{(\sum_{k = 1}^{N} p^{c*}_{kj}\lambda_k)^2}$ dictates that all $\tau^*_{ij}$, $v^*_{ij}$, and $\xi^*_i$ cannot simultaneously be zero. Furthermore, $h_i(\mathbf{P_C^*}) \neq 0$ implies that $\xi^*_i = 0$, in equation (\ref{kkt-cs3}). 
 Considering the sensing ability constraint for sensor $i$ in a feasible set, if it is not tight, $\xi^*_i$ must be $0$. This implies that $p^{c*}_{ij}$ is either $0$ or $1$, or else all $\tau^*_{ij}$, $v^*_{ij}$, and $\xi^*_i$ are $0$. In essence, if sensor $i$ cannot operate at maximum sensing ability, $p^{c*}_{ij}$ takes the values of either $0$ or $1$ for all $j \in [M]$. We can also say that the objective function gets smaller with the increase in any $p^{c*}_{ij}$ so that $p^{c*}_{ij}$ is set to $1$ instead of $0$ when the sensing ability constraint is not tight. We can conclude that every sensor $i$ uses its maximum sensing ability or $row_i(\mathbf{P_C^*}) = 1$ while the constraints are feasible. As a result, the convexity of the problem and the optimal solution greatly depends on the constraints. To that end, we define three different scenarios to analyze the minimum sum AoI. Particularly, a sensor's total sensing probability is contingent on the number of processes it tracks. We therefore consider three possible effects: (1) no impact, 
(2) improved total sensing probability with more processes tracked, and 
(3) weakened total sensing probability with more processes tracked. We note that the ways to achieve optimum sensing ability distribution are summarized in Remark 1. 
\begin{remark} \label{Rem1}
When the constraints in problem (\ref{minimization}) are convex, we can use KKT conditions to find the optimal correlation distribution. However, when the problem is non-convex, brute-force grid search algorithms may be used, although they suffer from high complexity. On the other hand, one can also leverage iterative algorithm such ADAM \cite{adam}, to find the solution to problem (\ref{minimization}). Although the algorithm does not necessarily converge to the global optimum, it still provides good performance in a low-complex fashion. For illustration purposes and analysis of the global optimum, we focus on a grid-search algorithm in the remainder of the paper. 
\end{remark}

Three example cases are presented to illustrate the possible regimes that can take place in the optimization. These include a linear constraint example, a quadratic convex constraint example, and a quadratic concave constraint example that correspond to no impact, improving impact, and weakening impact, respectively. For all three cases, we define a $\mathbf{b} = [b_1, ..., b_N]$ to reflect the intensity of correlation and sensors' sensing power.

\begin{enumerate}[wide, labelwidth=!, labelindent=0pt]
    \item \textbf{Linear Constraint Example:} \[\mathbf{P_C}\mathbf{1} -\mathbf{b} \leq 0.\]

This constraint indicates that there is neither a loss nor a gain to track more than one process and the total probability remains constant when probabilities are altered among different processes. If $b_i<1$, then the sensor $i$ is unable to consistently generate packets with updates. In addition, if $b_i>1$,  there are some intersections between processes. This sensor can be thought of as a camera that can track different areas and change its own position without any loss.  Let us consider the case where a camera may suffer from malfunctioning, thus hindering it from gathering information about the processes it monitors. This reflects a scenario where $b$ is diminishing. On the other hand, if the areas have some intersections, which represent a correlation between processes, or the camera has the ability to sense more than one process at the same time, this translates to the cases where $b_i$ is increasing. Example feasible sets are shown with blue area for a sensor with 2 processes in Figure \ref{linear}.
    \begin{figure}[!htbp]
\setlength{\belowcaptionskip}{-10pt}
\centering
\begin{minipage}[t]{0.33\linewidth}
\centering\includegraphics[height=0.8in]{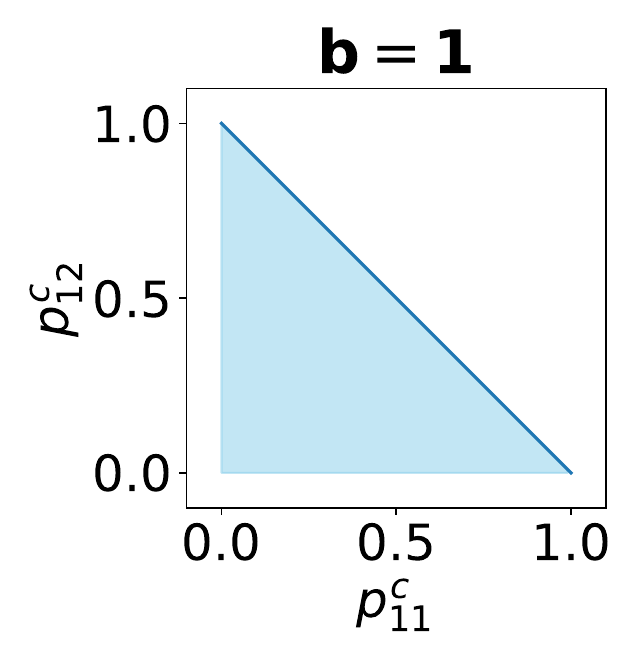}
\label{linear1}
\end{minipage}%
\begin{minipage}[t]{0.33\linewidth}
\centering\includegraphics[height=0.8in]{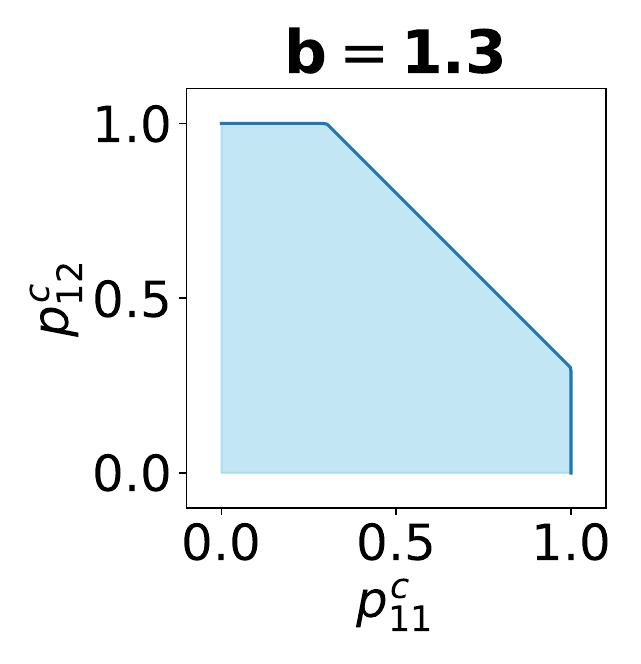}
\label{linear2}
\end{minipage}%
\begin{minipage}[t]{0.33\linewidth}
\centering\includegraphics[height=0.8in]{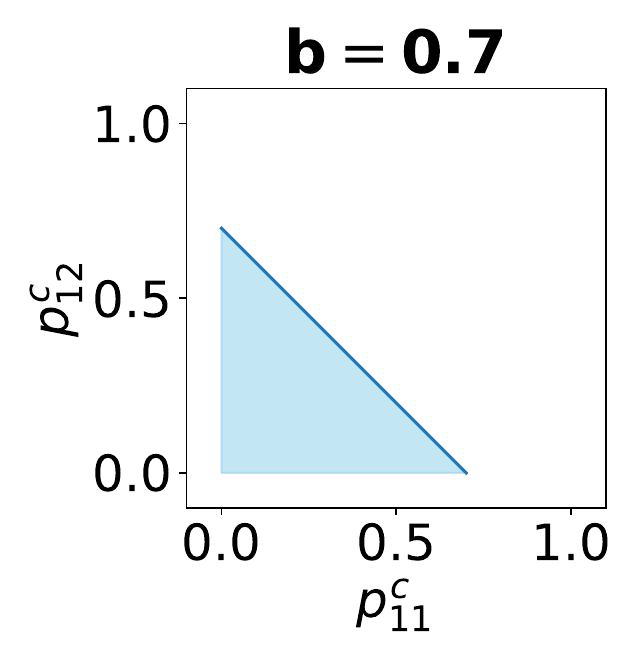}
 \label{linear3}
\end{minipage}

\caption{Example feasible sets of linear constraints with different $\mathbf{b}$ for sensor 1 with 2 processes.}
\label{linear}
\end{figure}

To minimize the sum AoI, it is crucial to allocate probabilities optimally. The objective presented in eq. (\ref{minimization}) is equivalent to maximizing the harmonic mean. It is also known that the harmonic mean is always smaller than or equal to the arithmetic mean, and they are only equal when all values are identical. In this example, we show that the arithmetic mean of the $\Tilde{p}^*$ is constant so that the optimal solution is distributing sensing probabilities to have equal $\Tilde{p}_j^*$ for all $j = 1,...,M$. This optimization process begins with obtaining sum of $\Tilde{p}_j^*$ for all $j = 1,...,M$, as follows, \[\sum_{j = 1}^{M}\Tilde{p_j} = \boldsymbol{\lambda^T}\mathbf{P_C}\mathbf{1}/\lambda_C \leq \boldsymbol{\lambda^T}\mathbf{b}.\] The KKT conditions imply that every sensor $i$ uses all sensing capabilities unless $M<b_i$. Otherwise, the value of the $i$-th row of $\mathbf{P_C}\mathbf{1}$ is $M$ when the row $i$ of the matrix $\mathbf{P_C}$ is equal to 1. Thus, we can say that $row_i(\mathbf{P_C}\mathbf{1}) = \min(M,b_i)$ and we can rewrite the summation as follows. \[\sum_{j = 1}^{M}\Tilde{p}_j^* = \boldsymbol{\lambda^T}\mathbf{P_C^*}\mathbf{1}/\lambda_C = \boldsymbol{\lambda^T}(\min(\mathbf{M},\textbf{b})).\]
Obtaining constant summation implies the arithmetic mean is also constant, and its value equals the maximum possible harmonic mean when all $\Tilde{p}_j^*$ are equal. This result implies that any $\mathbf{P_C}$ that satisfies conditions in sensing ability constraints is optimal when all $\Tilde{p}_j$ are equal. When $ N>1, M>1$, there might be more than one optimal solution. One possible solution that provides equal $\Tilde{p}_j$ is distributing sensing abilities equally, as follows:

\[p^{c*}_{ij} = \min(M,b_i)/M \frac{}{},\forall_{i \in [N], j \in [M]}.\]
This optimal solution can also be obtained from KKT conditions by letting $\frac{d}{dp^c_{ij}} h_k(\mathbf{P_C^*}) = 1$ and KKT conditions hold for this solution. 

    \item \textbf{Quadratic Convex Constraint Example:}
     \[row_i(\mathbf{P_{C}})row_i(\mathbf{P_{C}})^T - b_i \leq 0,  \quad i = 1, \dots, N.\]

We set $b_i$ values for intersections and sensor sensing abilities such that they can be less than or greater than 1 to represent the intersections and sensor sensing abilities, similar to the previous case. Then, the convex constraint states that the total probability may increase when probabilities are altered among different processes. To have a constraint that ensures the total probabilities increase when the sensor tracks more processes, let us consider a camera as a sensor that can track processes for different time periods. The camera can track a single process for a time period, and it can increase the probability of having information, but increasing the length of the time period may have a diminishing return. On the other hand, it is possible to split time and track two or more processes. While the probabilities of tracked processes decrease with an increase in the number of processes tracked, the sum of sensing probabilities increases compared to tracking a single process. Example feasible sets for a sensor with 2 processes are illustrated with blue area in Figure \ref{convex}.
\begin{figure}[!htbp]
\setlength{\belowcaptionskip}{-10pt}
\centering
\begin{minipage}[t]{0.33\linewidth}
\centering\includegraphics[height=0.8in]{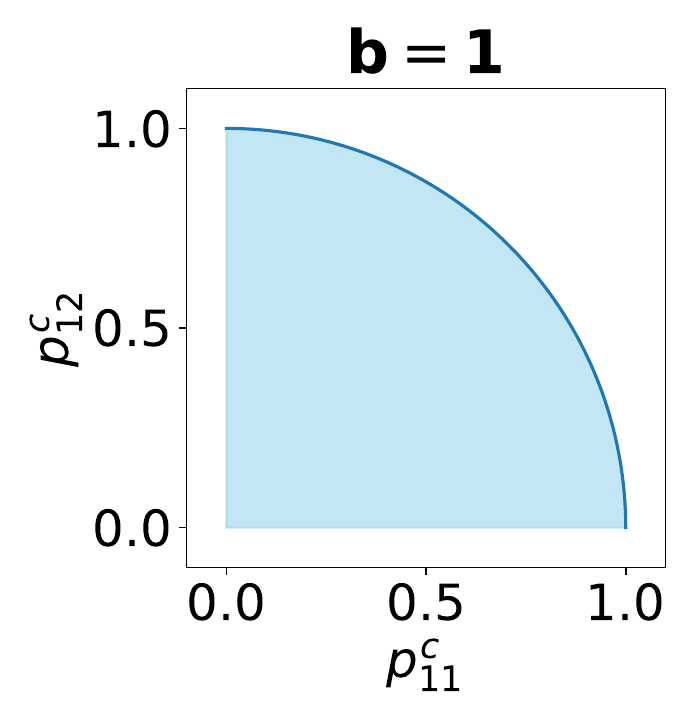}
\label{convex1}
\end{minipage}%
\begin{minipage}[t]{0.33\linewidth}
\centering\includegraphics[height=0.8in]{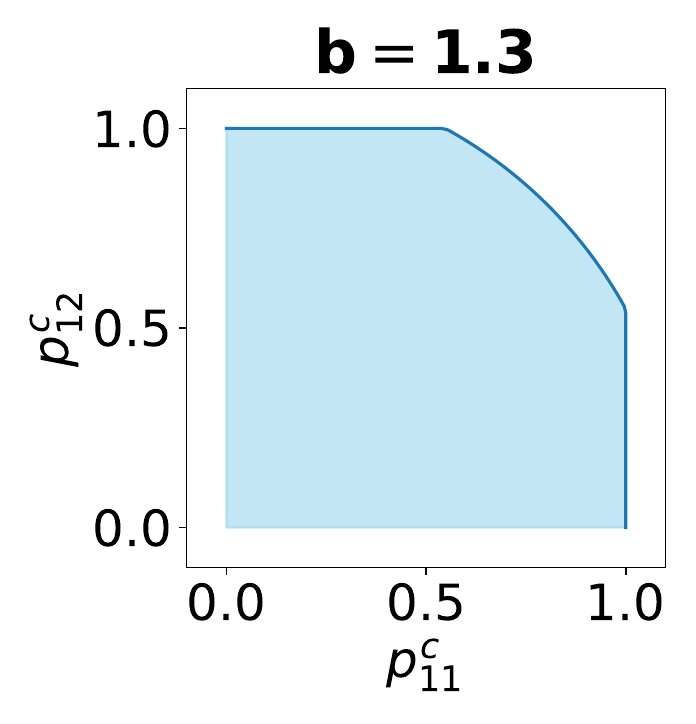}
\label{convex2}
\end{minipage}%
\begin{minipage}[t]{0.33\linewidth}
\centering\includegraphics[height=0.8in]{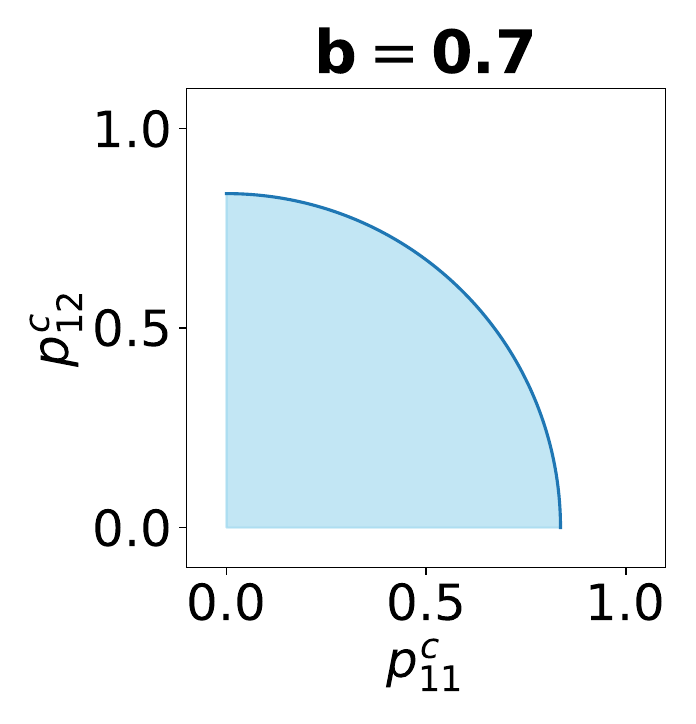}
\label{convex3}
\end{minipage}

\captionsetup{width=1\linewidth}
\caption{Example feasible sets of quadratic convex constraints with different $\mathbf{b}$ for sensor 1 with 2 processes.} 
\label{convex}

\end{figure}

Similarly to the previous case, the sensor may have the ability to track more processes than $M$ and could track all processes, but if unable, it uses all of its sensing ability to reach the optimum, which can be written as $row_i(\mathbf{P_{C}})row_i(\mathbf{P_{C}})^T = \min(M,b_i)$. If the sensor has the ability to track all processes $(b_i \geq M)$ then all probability values for the sensor are $1$, $(p^{c*}_{ij} = 1 \frac{}{},\forall_{j \in [M]})$. For the sensors that can't track all processes at the same time $(b_i < M)$, we use KKT conditions. Let us consider sensor $i$ and assume $p^{c*}_{ij} \neq 1$ or $p^{c*}_{ij} \neq 0$ for all $j \in [M]$, we find that $v^*_{ij} = \tau^*_{ij} = 0$ for all $j \in [M]$. Additionally, the derivative of $h_k(\mathbf{P_C^*})$ with respect to $p^c_{ij}$ is $2p^{c*}_{ij}$. Substituting these values into equation (\ref{KKT-derivative2}), we obtain the condition below.
\[\frac{\lambda_i}{\lambda_C \xi^*_i} =  2p^{c*}_{ij}\Tilde{p}_j^{*2}, \quad i = 1, \dots, N, \quad j = 1, \dots, M.\]
When we keep $i$ constant, we have the same left-hand side, so the optimal solution must satisfy these conditions. In this case, a potential solution could be to distribute the sensing ability of the sensor equally such that,
\[p^{c*}_{ij} = \sqrt{\min(M,b_i)/M} \frac{}{}, \quad i = 1, \dots, N, \quad j = 1, \dots, M, \]  and KKT conditions hold for this solution.

    \item \textbf{Quadratic Concave Constraint Example:} \[b_i - row_i(\mathbf{1-P_{C}})row_i(\mathbf{1-P_{C}})^T  \leq 0,  \quad i = 1, \dots, N.\]
We assign $b_i$ values to intersections and sensor sensing ability, which can be less than or greater than 1 to represent them, similar to previous cases. This constraint indicates that the overall probability may decrease when probabilities are changed across different processes, rather than tracking a single process. Let us consider an example where a camera can track more than one process, but changing the camera's pose takes some time, during which it cannot generate updates with information. In this scenario, tracking more processes increases the total probability of losses, and the total probability is higher when the camera focuses only on a single process. However, if there is no update from the other processes, its AoI goes to infinity, which is also an undesired distribution. This trade-off decides the optimal distribution. Figure \ref{concave} illustrates the feasible sets for a sensor that has 2 processes with blue area.
            \begin{figure}[!htbp]
\setlength{\belowcaptionskip}{-10pt}
\centering
\captionsetup{width=0.3\linewidth}
\begin{minipage}[t]{0.33\linewidth}
\centering\includegraphics[height=0.8in]{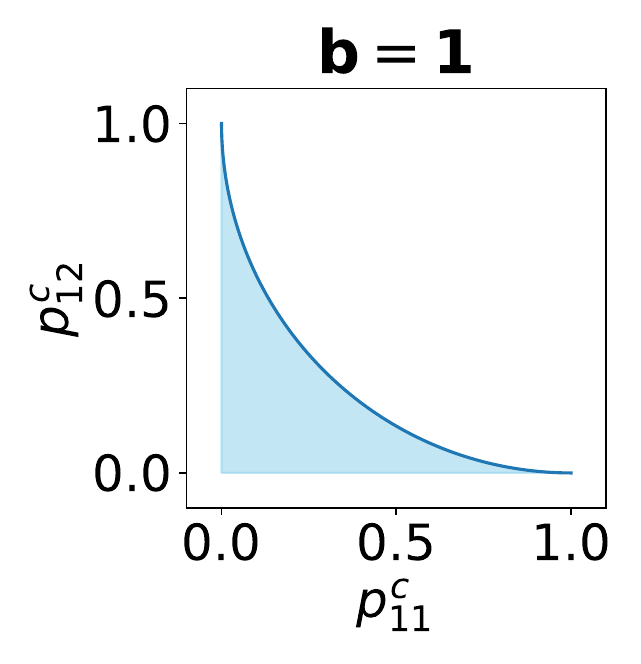}
 \label{concave1}
\end{minipage}%
\begin{minipage}[t]{0.33\linewidth}
\centering\includegraphics[height=0.8in]{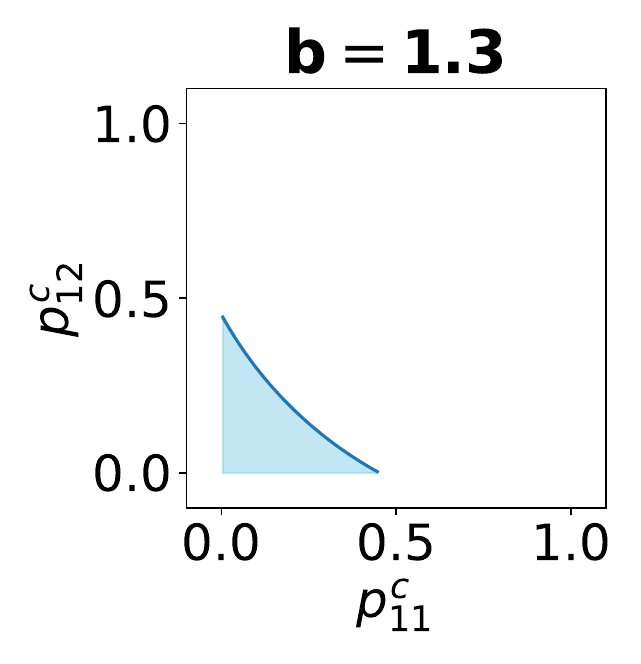}
\label{concave2}
\end{minipage}%
\begin{minipage}[t]{0.33\linewidth}
\centering\includegraphics[height=0.8in]{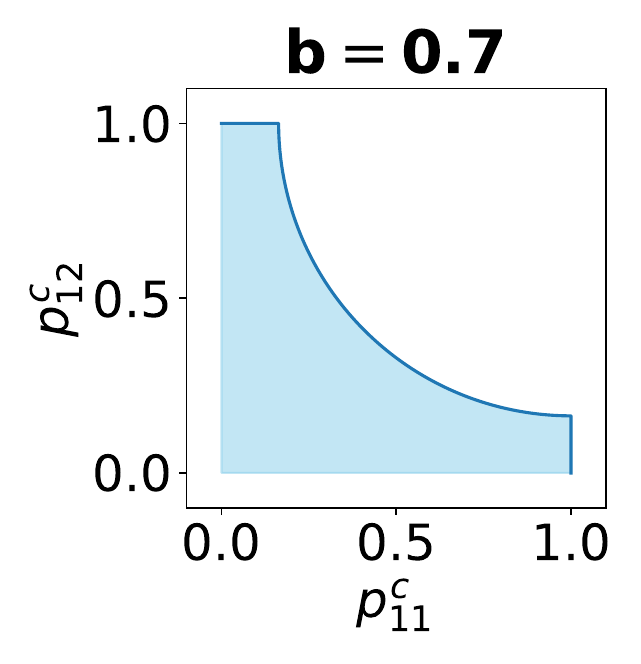}
\label{concave3}
\end{minipage}
\captionsetup{width=1\linewidth}

\caption{Example feasible sets of quadratic concave constraints with different $\mathbf{b}$ for sensor 1 with 2 processes.}
\label{concave}

\end{figure}

In this case, we end up with a non-convex set, unlike the previous two cases. As a result, our optimization problem is no longer convex. Although equal distribution satisfies KKT conditions, it may not be the global optimum. Therefore, non-convex optimization approaches, such as ADAM optimizer \cite{adam}, can be applied in such cases. In the next section, first, we verify our theoretical analysis by presenting numerical results, and then we present optimal $\mathbf{P_C}$ distributions for different scenarios.
\end{enumerate}	
\section{Numerical Results}\label{numerical-s}
We categorize our numerical analysis into two groups. First, we assume each sensor is assigned to a process, which means
every sensor has a target process and every packet from this sensor contains information about the target process. In this case, the packet may also contain information about other sources. We vary different system parameters to verify our theoretical results in Section \ref{aoi-S} and Section \ref{error-s}. After verifying the theoretical analysis, we investigate the trade-off model described in Section \ref{opt-s} such that each sensor has some sensing abilities and constraints. In this case, the sensors can distribute their sensing abilities over processes, and we investigate the best $\mathbf{P_C}$ distribution to get the minimum sum AoI.
\begin{figure*}[!htbp]
\setlength{\belowcaptionskip}{-10pt}
\centering
\captionsetup{width=0.28\linewidth}
\begin{minipage}[t]{0.32\linewidth}
\centering\includegraphics[height=1.1in]{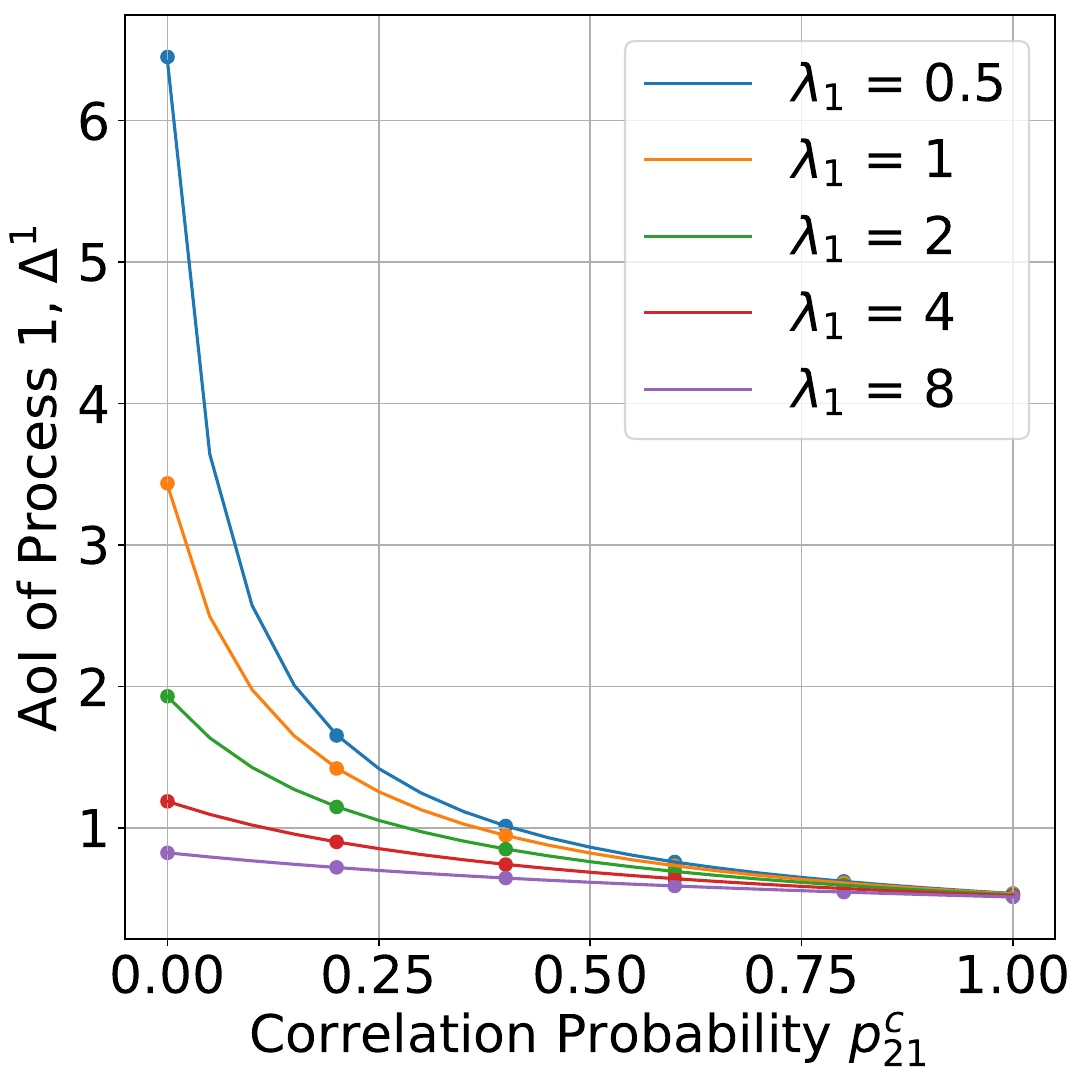}
	\caption{AoI versus $p^c_{21}$ for different $\lambda_i$ values with $\mu=4, \zeta_{1}=4,$ $\zeta_{2}=4, \lambda_2 = 8$. } \label{aoi_fig}
\end{minipage}%
\begin{minipage}[t]{0.32\linewidth}
\centering\includegraphics[height=1.1in]{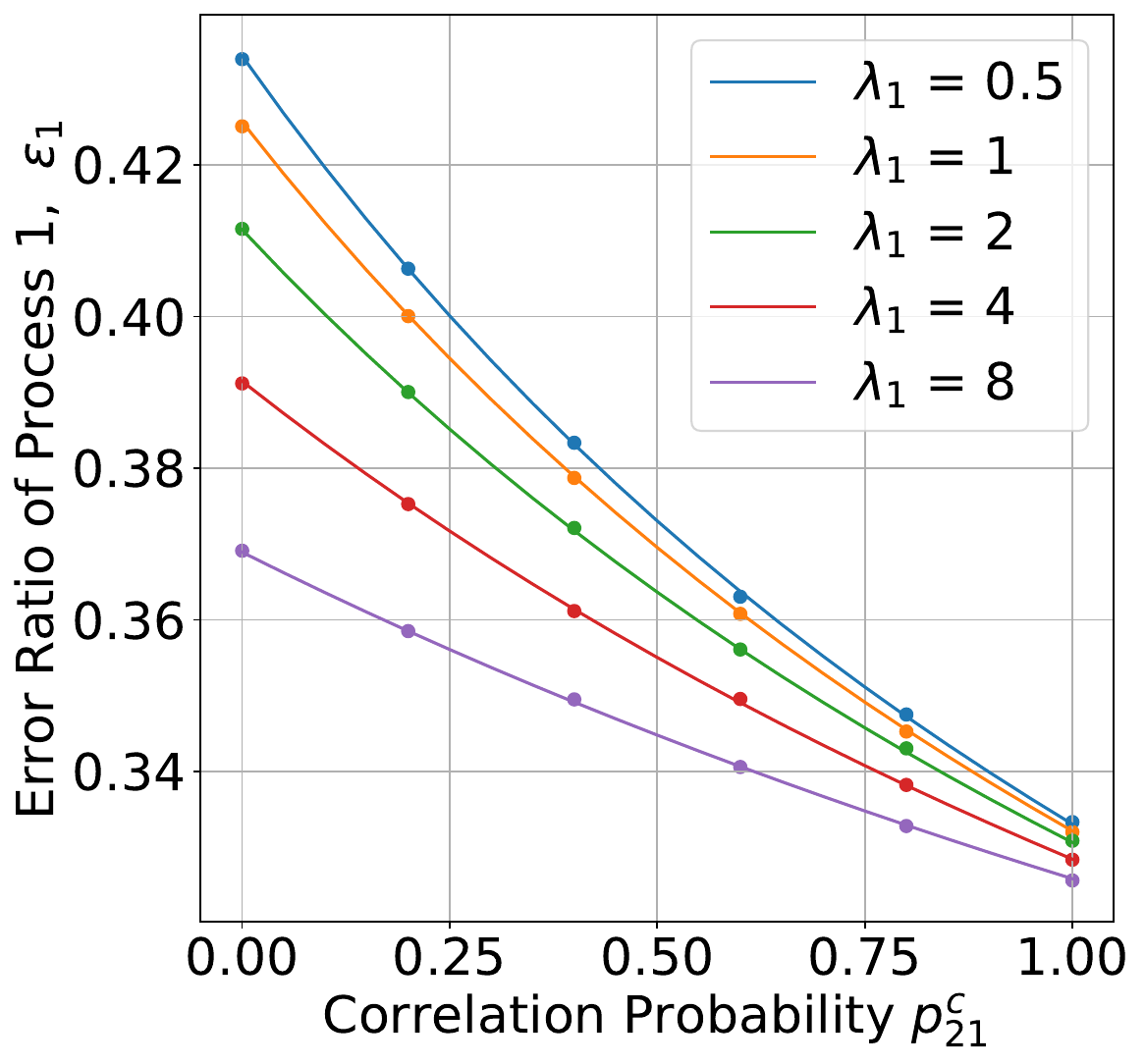}
	\caption{Error $\epsilon_1$ versus $p^c_{21}$ for different $\lambda_i$ values with $\mu=4, \lambda_2 = 8$.} \label{error_fig_corr}
\end{minipage}
\begin{minipage}[t]{0.32\linewidth}
\centering\includegraphics[height=1.1in]{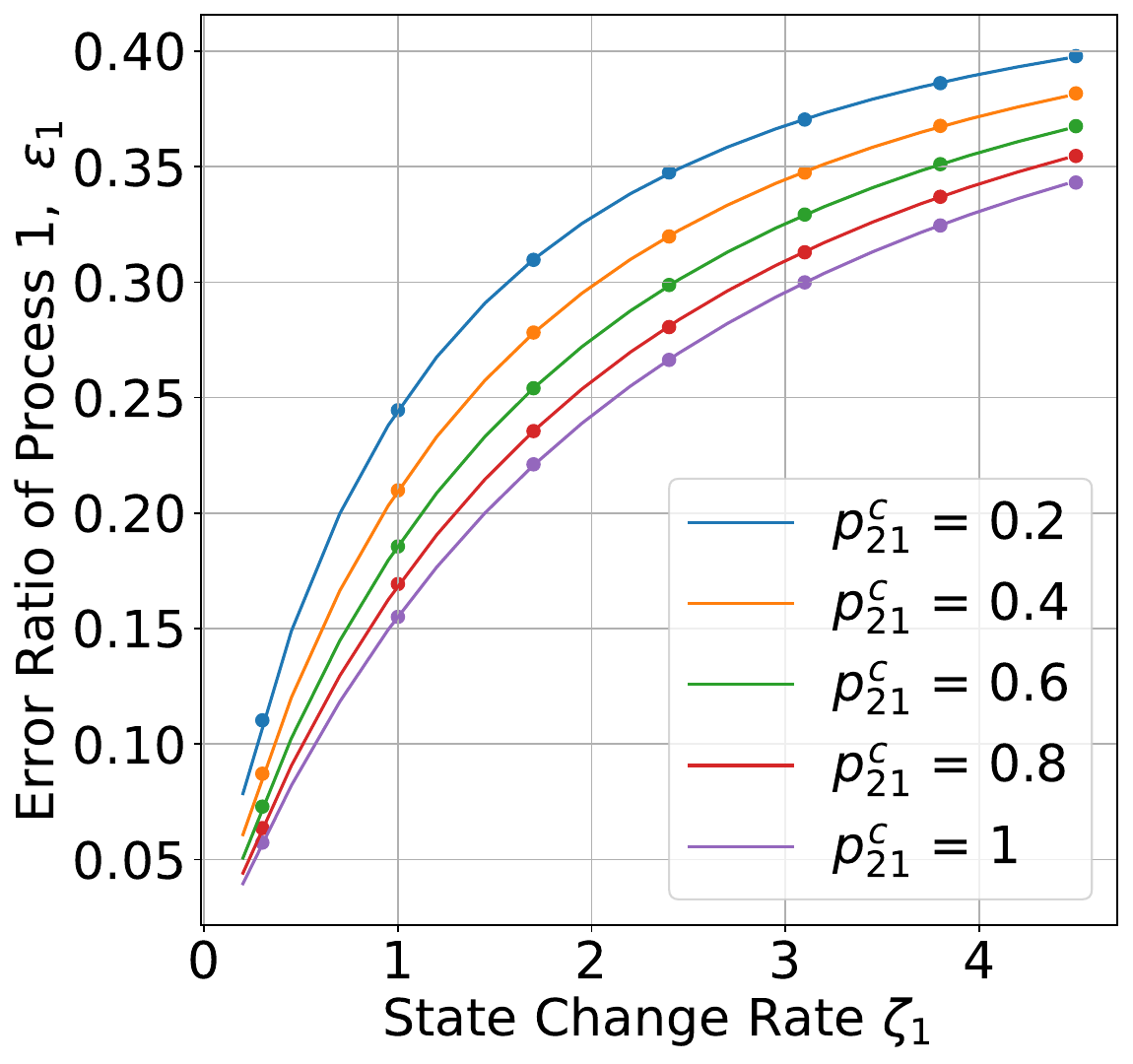}
	\caption{Error $\epsilon_1$ versus $\zeta_{1}$ for different $p^c_{21}$ values with $\mu=4, \zeta_{2}=4, \lambda_1 = 2, \lambda_2 = 8$. 
 } \label{error_fig_state}
\end{minipage}
\end{figure*}

\subsection{Sensor-Process Assignment Analysis}
In this section, we provide numerical results for the average AoI and error ratio in a system with two processes having two states along with two sensors to see the effects of correlation probability. Our simulations are unit-time-based and were run for 1 million units of time. The lowest arrival rate is 0.5 arrivals per unit time. This means we have at least $5\times10^5$ arrivals for each process to guarantee convergence. Simulation results are shown as circles, and theoretical expressions from our analysis are shown as solid lines in the figures. The strong match between them verifies the validity of our analysis. We have three figures that depict different relationships. The first two show the effect of correlation on AoI and error with varying arrival rates. The third one shows how the error changes with the state change rate for different correlations. We set $\boldsymbol{\Omega_1}$, $\boldsymbol{\Omega_2}$ and $\mathbf{P_C}$ as follows for the three cases. 
\[
\boldsymbol{\Omega_1} = \boldsymbol{\Omega_2} = \begin{bmatrix}
0.4 & 0.6 \\
0.3 & 0.7
\end{bmatrix}, \quad \boldsymbol{\mathbf{P_C}} = \begin{bmatrix}
1 & 1-p \\
1-p & 1
\end{bmatrix},
\]
where $p$ is a tuning variable to vary $p^c_{21}$ and $p^c_{12}$. As a result of our analysis, only $col_j(\mathbf{P_C})$ has an effect on metrics for process $j$. Therefore, we only focus on $col_1(\mathbf{P_C})$ in this part.

We start with investigating the effects of correlation on AoI in Figure \ref{aoi_fig}. We vary the correlation probability  $p^c_{21}$ and $\lambda_1$ with a service rate $\mu$ fixed to 4, state change rate for process 1 and process 2, $\zeta_{1}$ and $\zeta_{2}$ fixed to 4, the arrival rate for sensor 2, $\lambda_2$ is fixed to 8. As $p^c_{21}$ increases, AoI decreases as expected. However, the impact of correlation gets smaller with the increase in sensor 1's arrival rate $\lambda_1$. The reason is that an increase in $\lambda^*_1$ causes a diminishing AoI drop, and correlation has a small impact when the $\lambda_1$ is high enough.  As $\lambda_1$ increases, the status updates become more frequent, but there is another limitation on AoI, which is service time. Therefore, the AoI converges while $\lambda_i$ is increasing. However, if $\lambda_1$ is small, and there is another sensor with a high arrival rate $\lambda_2$, correlation plays a huge role in AoI.

Next, we observe the error ratio for the same system in Figure \ref{error_fig_corr}. We vary the correlation probability $p^c_{21}$ and $\lambda_1$ with a service rate $\mu$ fixed to 4, state change rate for process 1 and process 2, $\zeta_{1}$ and $\zeta_{2}$ fixed to 4, arrival rate for sensor 2, $\lambda_2$ fixed to 8. We see a very similar pattern to Figure \ref{aoi_fig} in Figure \ref{error_fig_corr}. As $p^c_{21}$ increases, the error drops as expected because of an increase in the ratio of informative packets in the server. Correlation has a powerful effect on error ratio, unlike AoI for $\lambda_i = 8$. As $p^c_{21}$, the difference between errors of different $\lambda_i$ values decreases because all packets become informative. However, an increase in the arrival rate causes convergence; in other words, an error-less system cannot be achieved because both $\mu$ and $\zeta_{1}$ are equal to 4. State changes during service time cause errors, and the errors causing them cannot be reduced by either increasing the arrival rate or increasing correlation.

Lastly, we investigate the system with a service rate $\mu$ fixed to 4, arrival rate for sensor 1, $\lambda_1$ fixed to 2, arrival rate for sensor 2, $\lambda_2$ fixed to 8. We vary the correlation probability  $p^c_{21}$ and $\zeta_{1}$. We only show the error ratio for this system in Figure \ref{error_fig_state} because changes in $\zeta_{1}$ do not have any effect on AoI. As $\zeta_{1}$ increases, the error increases as expected because the state changes more frequently. For small $\zeta_{1}$, the error is almost zero for all correlation values. The error ratio converges to  1 - $\inner{\boldsymbol{\psi}}{\boldsymbol{\psi}}$, where $\inner{\boldsymbol{\psi}}{\boldsymbol{\psi}}$ represents the inner product of stationary distribution $\boldsymbol{\psi}$ defined in (\ref{stationary_state}), as $\zeta_{1}$ goes to $\infty$ because $\lambda_1$ and $\lambda_2$ lose their effect when $\zeta_{1}$ is too large. Although the start and end points of all plots for error are quite similar, when the correlation is high, the increase becomes slower, which makes the system more robust. 
\subsection{Sensing Ability-Constrained Sensor Optimization Model}
To evaluate the optimal correlation distribution, we consider a system with two sensors and two processes. We vary $\lambda_2$ from $1$ to $100$ and set $\lambda_1 = 1, \mu = 4, \zeta_{1}=0.4,  \zeta_{2}=0.4 $. We also set $\boldsymbol{\Omega_1}, \boldsymbol{\Omega_2}$ and $\mathbf{b}$ as follows: 
\[
\boldsymbol{\Omega_1} = \boldsymbol{\Omega_2} = \small\begin{bmatrix}
0.4 & 0.6 \\
0.3 & 0.7
\end{bmatrix}\normalsize, \quad \mathbf{b} = \begin{bmatrix} 1 & 1 \end{bmatrix}.
\]Then, we use the SLSQP algorithm \cite{slsqp} for convex problems to converge to the global optimum. For the non-convex part, we use grid search to find the distribution that provides the global minimum AoI to avoid converging to a local minimum. We use a step size of $10^{-3}$ for each probability variable.
We use the three constraints defined in Section \ref{opt-s}, and we obtain that the constraints are always tight when $ \mathbf{b} \leq M$ and sensors use their maximum sensing abilities, which verifies our argument derived from KKT conditions. The other result we see is that an equal distribution of the sensors' sensing probabilities among processes is optimal for the first two constraints.

After verifying the first two constraints, we evaluate the third constraint that makes the problem non-convex. The results are shown in Figure \ref{pc21} and we only present $p^{c*}_{21}$ and $p^{c*}_{22}$ for simplicity, since we obtain $p^{c*}_{11}=1$ and $p^{c*}_{12} = 0$ for all $\lambda_2$. Note that switching $p^{c*}_{11}$ and $p^{c*}_{21}$ to $p^{c*}_{12}$ and $p^{c*}_{22}$, respectively, results in the same AoI value due to the symmetry of the case. Consequently, we observe equal sensing probabilities are not always the best distribution, in contrast to the previous cases. When sensor $\lambda$'s are close to each other, assigning a sensor to a process provides the minimum AoI. However, if one sensor dominates the other in terms of arrival rate, the equal distribution provides the optimal AoI for the dominating sensor. Interestingly, we observe that the value of $p^{c*}_{21}$ does not change smoothly from $1$. We obtain $p^{c*}_{21} = 0 $ if $\lambda_2 \lesssim  3.18$ but we see a dramatic regime switch to $0.08$ after around $3.18$. The reason behind this is tracking more than one process causes a loss of total sensing ability. Therefore, if $\lambda$'s are close to each other, there is a tendency to track a single process from each sensor. Nevertheless, if one $\lambda$ is significantly larger than the other, the system will converge to the configuration with a single sensor. In order to avoid infinite AoI values, both processes must be tracked by the dominating sensor. 

Lastly, in addition to AoI minimization, we also consider the error minimization problem. Interestingly, the same trend can be observed in the error minimization problem. The optimal assignment involves assigning a process to the sensor until it reaches a $\lambda_2$ threshold, albeit different from the AoI $\lambda_2$ threshold. After the threshold, it converges to equal distribution while $\lambda_2$ is increasing. Note that, although only $\mathbf{b}$ and $\lambda_1$ affect the $\lambda_2$ threshold for AoI, other parameters like $\Omega_1$, $\Omega_2$, $\zeta_{1}$, $\zeta_{2}$, and $\mu$ also have an effect on the $\lambda_2$ threshold for error.
\begin{figure}[!htbp]
\setlength{\belowcaptionskip}{-10pt}
\centering

\begin{subfigure}[t]{0.5\linewidth}
\captionsetup{width=0.85\linewidth}
\centering
\includegraphics[height=1in]{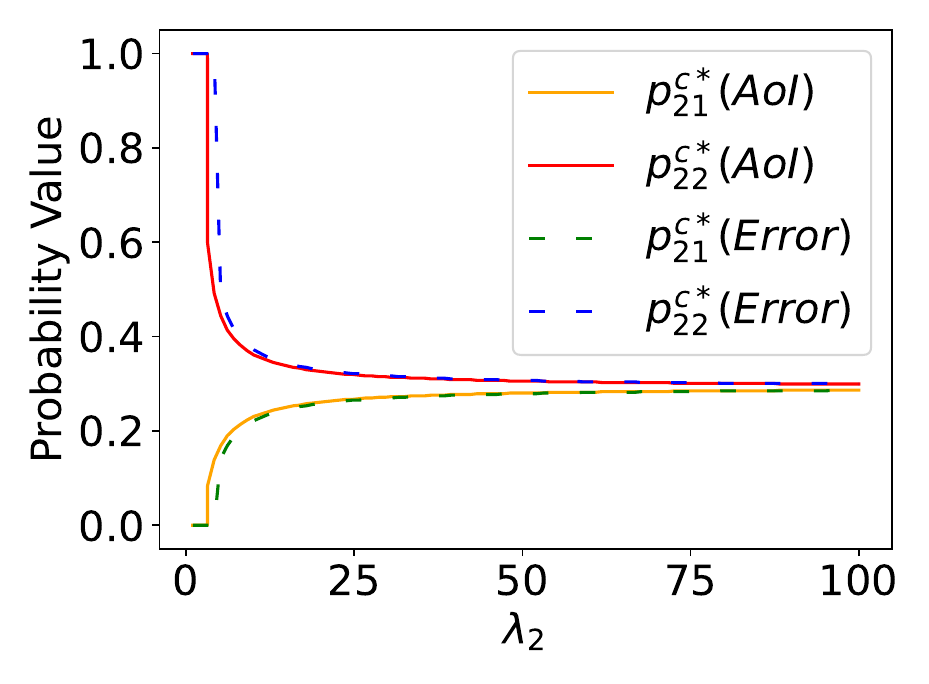}
\caption{$\lambda_2$ versus optimal $p^{c*}_{21}$ and $p^{c*}_{22}$ value with $\lambda_1 = 1, \mu = 4, \mathbf{b} = 1, \zeta_{1}=0.4,\zeta_{2}=0.4$.}
\label{pc21a}
\end{subfigure}%
\begin{subfigure}[t]{0.5\linewidth}
\captionsetup{width=0.85\linewidth}
\centering\includegraphics[height=1in]{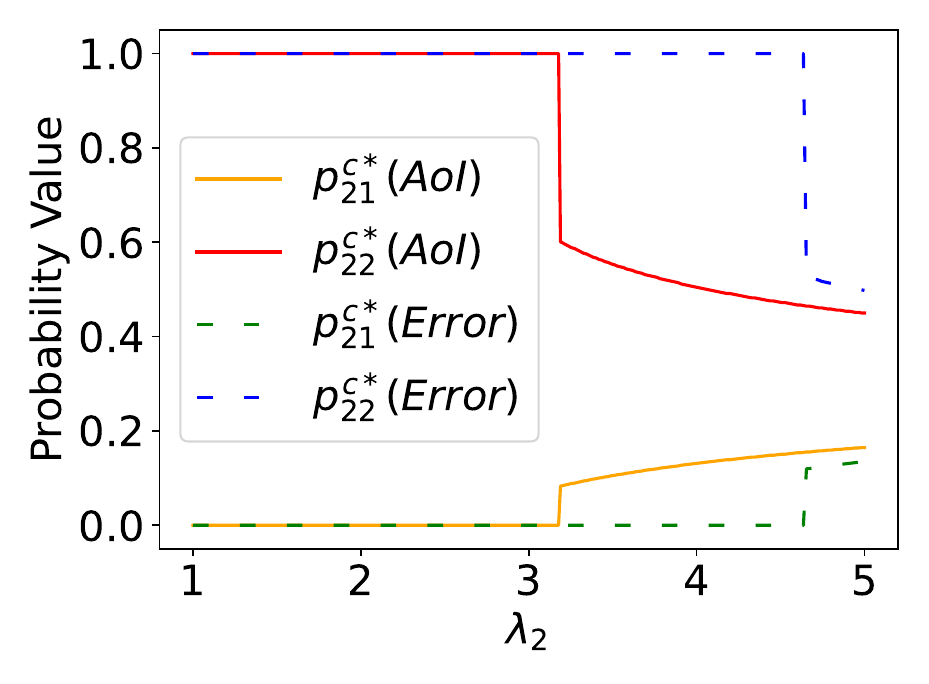}
\caption{A closer look at the regime switch in Figure a.}
\label{pc21b}
\end{subfigure}

\caption{Optimal distribution for quadratic non-convex distribution example.}
\label{pc21}
\end{figure}

\section{Conclusion}\label{conc-s}

Our paper investigates multi-process multi-sensor systems with a single server and correlated information. The study focuses on understanding the effects of correlation on the system, particularly in terms of the age of information and the error ratio of the monitor's estimation. We present analytical expressions and conducted comparisons to evaluate the influence of correlation across different system parameters. Our work also delves into how sensors with restricted sensing abilities should allocate resources across processes. We have highlighted the links between correlation, processes' state change rate, and optimization tactics for managing the age of information. Interestingly, we identified cases where monitoring multiple processes from one source may not always be beneficial, and optimal resource distribution for different arrival rates may exhibit critical regime shift behavior.

\section*{Acknowledgment}
This work was supported in part by NSF grant CNS-2219180.

\bibliographystyle{ACM-Reference-Format}
\bibliography{references}


\begin{thebibliography}{36}


\ifx \showCODEN    \undefined \def \showCODEN     #1{\unskip}     \fi
\ifx \showDOI      \undefined \def \showDOI       #1{#1}\fi
\ifx \showISBNx    \undefined \def \showISBNx     #1{\unskip}     \fi
\ifx \showISBNxiii \undefined \def \showISBNxiii  #1{\unskip}     \fi
\ifx \showISSN     \undefined \def \showISSN      #1{\unskip}     \fi
\ifx \showLCCN     \undefined \def \showLCCN      #1{\unskip}     \fi
\ifx \shownote     \undefined \def \shownote      #1{#1}          \fi
\ifx \showarticletitle \undefined \def \showarticletitle #1{#1}   \fi
\ifx \showURL      \undefined \def \showURL       {\relax}        \fi
\providecommand\bibfield[2]{#2}
\providecommand\bibinfo[2]{#2}
\providecommand\natexlab[1]{#1}
\providecommand\showeprint[2][]{arXiv:#2}

\bibitem[Arafa et~al\mbox{.}(2020)]%
        {const-arafa}
\bibfield{author}{\bibinfo{person}{Ahmed Arafa}, \bibinfo{person}{Jing Yang},
  \bibinfo{person}{Sennur Ulukus}, {and} \bibinfo{person}{H.~Vincent Poor}.}
  \bibinfo{year}{2020}\natexlab{}.
\newblock \showarticletitle{Age-Minimal Transmission for Energy Harvesting
  Sensors With Finite Batteries: Online Policies}.
\newblock \bibinfo{journal}{\emph{IEEE Transactions on Information Theory}}
  \bibinfo{volume}{66}, \bibinfo{number}{1} (\bibinfo{year}{2020}),
  \bibinfo{pages}{534--556}.
\newblock
\urldef\tempurl%
\url{https://doi.org/10.1109/TIT.2019.2938969}
\showDOI{\tempurl}


\bibitem[Bacinoglu et~al\mbox{.}(2015)]%
        {const-biyikoglu}
\bibfield{author}{\bibinfo{person}{Baran~Tan Bacinoglu},
  \bibinfo{person}{Elif~Tugce Ceran}, {and} \bibinfo{person}{Elif
  Uysal-Biyikoglu}.} \bibinfo{year}{2015}\natexlab{}.
\newblock \showarticletitle{Age of information under energy replenishment
  constraints}. In \bibinfo{booktitle}{\emph{2015 Information Theory and
  Applications Workshop (ITA)}}. \bibinfo{pages}{25--31}.
\newblock
\urldef\tempurl%
\url{https://doi.org/10.1109/ITA.2015.7308962}
\showDOI{\tempurl}


\bibitem[Baknina et~al\mbox{.}(2018)]%
        {const-ulukus}
\bibfield{author}{\bibinfo{person}{Abdulrahman Baknina}, \bibinfo{person}{Omur
  Ozel}, \bibinfo{person}{Jing Yang}, \bibinfo{person}{Sennur Ulukus}, {and}
  \bibinfo{person}{Aylin Yener}.} \bibinfo{year}{2018}\natexlab{}.
\newblock \showarticletitle{Sending Information Through Status Updates}. In
  \bibinfo{booktitle}{\emph{2018 IEEE International Symposium on Information
  Theory (ISIT)}}. \bibinfo{pages}{2271--2275}.
\newblock
\urldef\tempurl%
\url{https://doi.org/10.1109/ISIT.2018.8437496}
\showDOI{\tempurl}


\bibitem[Bedewy et~al\mbox{.}(2016)]%
        {bedewy}
\bibfield{author}{\bibinfo{person}{Ahmed~M. Bedewy}, \bibinfo{person}{Yin Sun},
  {and} \bibinfo{person}{Ness~B. Shroff}.} \bibinfo{year}{2016}\natexlab{}.
\newblock \showarticletitle{Optimizing data freshness, throughput, and delay in
  multi-server information-update systems}. In \bibinfo{booktitle}{\emph{2016
  IEEE International Symposium on Information Theory (ISIT)}}.
  \bibinfo{pages}{2569--2573}.
\newblock
\urldef\tempurl%
\url{https://doi.org/10.1109/ISIT.2016.7541763}
\showDOI{\tempurl}


\bibitem[Bertsekas and Gallager(1996)]%
        {dataNetworks:book}
\bibfield{author}{\bibinfo{person}{Dimitri Bertsekas} {and}
  \bibinfo{person}{Robert Gallager}.} \bibinfo{year}{1996}\natexlab{}.
\newblock \bibinfo{booktitle}{\emph{Data Networks} (\bibinfo{edition}{second}
  ed.)}.
\newblock \bibinfo{publisher}{Prentice Hall}.
\newblock


\bibitem[Boyd and Vandenberghe(2004)]%
        {boyd2004convex}
\bibfield{author}{\bibinfo{person}{Stephen Boyd} {and} \bibinfo{person}{Lieven
  Vandenberghe}.} \bibinfo{year}{2004}\natexlab{}.
\newblock \bibinfo{booktitle}{\emph{Convex optimization}}.
\newblock \bibinfo{publisher}{Cambridge university press}.
\newblock


\bibitem[C.C.~Paige and †(1975)]%
        {eigenvalue}
\bibfield{author}{\bibinfo{person}{George P.H.~Styan C.C.~Paige} {and}
  \bibinfo{person}{Peter G.~Wachter †}.} \bibinfo{year}{1975}\natexlab{}.
\newblock \showarticletitle{{Computation of the stationary distribution of a
  markov chain}}.
\newblock \bibinfo{journal}{\emph{Journal of Statistical Computation and
  Simulation}} \bibinfo{volume}{4}, \bibinfo{number}{3} (\bibinfo{year}{1975}),
  \bibinfo{pages}{173--186}.
\newblock
\urldef\tempurl%
\url{https://doi.org/10.1080/00949657508810122}
\showDOI{\tempurl}
\showeprint{https://doi.org/10.1080/00949657508810122}


\bibitem[Costa et~al\mbox{.}(2016)]%
        {mm1}
\bibfield{author}{\bibinfo{person}{Maice Costa}, \bibinfo{person}{Marian
  Codreanu}, {and} \bibinfo{person}{Anthony Ephremides}.}
  \bibinfo{year}{2016}\natexlab{}.
\newblock \showarticletitle{{On the Age of Information in Status Update Systems
  With Packet Management}}.
\newblock \bibinfo{journal}{\emph{IEEE Transactions on Information Theory}}
  \bibinfo{volume}{62}, \bibinfo{number}{4} (\bibinfo{year}{2016}),
  \bibinfo{pages}{1897--1910}.
\newblock
\urldef\tempurl%
\url{https://doi.org/10.1109/TIT.2016.2533395}
\showDOI{\tempurl}


\bibitem[Farazi et~al\mbox{.}(2018)]%
        {const-farazi}
\bibfield{author}{\bibinfo{person}{Shahab Farazi}, \bibinfo{person}{Andrew~G.
  Klein}, {and} \bibinfo{person}{D.~Richard Brown}.}
  \bibinfo{year}{2018}\natexlab{}.
\newblock \showarticletitle{Average age of information for status update
  systems with an energy harvesting server}. In \bibinfo{booktitle}{\emph{IEEE
  INFOCOM 2018 - IEEE Conference on Computer Communications Workshops (INFOCOM
  WKSHPS)}}. \bibinfo{pages}{112--117}.
\newblock
\urldef\tempurl%
\url{https://doi.org/10.1109/INFCOMW.2018.8406846}
\showDOI{\tempurl}


\bibitem[Gallager(1996)]%
        {splitting_poisson}
\bibfield{author}{\bibinfo{person}{Robert~G. Gallager}.}
  \bibinfo{year}{1996}\natexlab{}.
\newblock \bibinfo{booktitle}{\emph{{{Poisson Processes}}}}.
\newblock \bibinfo{publisher}{Springer US}, \bibinfo{address}{Boston, MA},
  \bibinfo{pages}{31--55}.
\newblock
\showISBNx{978-1-4615-2329-1}
\urldef\tempurl%
\url{https://doi.org/10.1007/978-1-4615-2329-1_2}
\showDOI{\tempurl}


\bibitem[He et~al\mbox{.}(2018)]%
        {he2018}
\bibfield{author}{\bibinfo{person}{Qing He}, \bibinfo{person}{Gyorgy Dan},
  {and} \bibinfo{person}{Viktoria Fodor}.} \bibinfo{year}{2018}\natexlab{}.
\newblock \showarticletitle{{Minimizing age of correlated information for
  wireless camera networks}}. In \bibinfo{booktitle}{\emph{IEEE INFOCOM 2018 -
  IEEE Conference on Computer Communications Workshops (INFOCOM WKSHPS)}}.
  \bibinfo{pages}{547--552}.
\newblock
\urldef\tempurl%
\url{https://doi.org/10.1109/INFCOMW.2018.8406914}
\showDOI{\tempurl}


\bibitem[He et~al\mbox{.}(2022)]%
        {collaborativesensing}
\bibfield{author}{\bibinfo{person}{Shibo He}, \bibinfo{person}{Kun Shi},
  \bibinfo{person}{Chen Liu}, \bibinfo{person}{Bicheng Guo},
  \bibinfo{person}{Jiming Chen}, {and} \bibinfo{person}{Zhiguo Shi}.}
  \bibinfo{year}{2022}\natexlab{}.
\newblock \showarticletitle{Collaborative Sensing in Internet of Things: A
  Comprehensive Survey}.
\newblock \bibinfo{journal}{\emph{IEEE Communications Surveys \& Tutorials}}
  \bibinfo{volume}{24}, \bibinfo{number}{3} (\bibinfo{year}{2022}),
  \bibinfo{pages}{1435--1474}.
\newblock
\urldef\tempurl%
\url{https://doi.org/10.1109/COMST.2022.3187138}
\showDOI{\tempurl}


\bibitem[Hsu et~al\mbox{.}(2017)]%
        {modiano-sch-1}
\bibfield{author}{\bibinfo{person}{Yu-Pin Hsu}, \bibinfo{person}{Eytan
  Modiano}, {and} \bibinfo{person}{Lingjie Duan}.}
  \bibinfo{year}{2017}\natexlab{}.
\newblock \showarticletitle{{Age of information: Design and analysis of optimal
  scheduling algorithms}}. In \bibinfo{booktitle}{\emph{2017 IEEE International
  Symposium on Information Theory (ISIT)}}. \bibinfo{pages}{561--565}.
\newblock
\urldef\tempurl%
\url{https://doi.org/10.1109/ISIT.2017.8006590}
\showDOI{\tempurl}


\bibitem[Huang and Modiano(2015)]%
        {modiano2015}
\bibfield{author}{\bibinfo{person}{Longbo Huang} {and} \bibinfo{person}{Eytan
  Modiano}.} \bibinfo{year}{2015}\natexlab{}.
\newblock \showarticletitle{{Optimizing age-of-information in a multi-class
  queueing system}}. In \bibinfo{booktitle}{\emph{2015 IEEE International
  Symposium on Information Theory (ISIT)}}. \bibinfo{pages}{1681--1685}.
\newblock
\urldef\tempurl%
\url{https://doi.org/10.1109/ISIT.2015.7282742}
\showDOI{\tempurl}


\bibitem[Kadota et~al\mbox{.}(2018)]%
        {sch-igor-1}
\bibfield{author}{\bibinfo{person}{Igor Kadota}, \bibinfo{person}{Abhishek
  Sinha}, \bibinfo{person}{Elif Uysal-Biyikoglu}, \bibinfo{person}{Rahul
  Singh}, {and} \bibinfo{person}{Eytan Modiano}.}
  \bibinfo{year}{2018}\natexlab{}.
\newblock \showarticletitle{{Scheduling Policies for Minimizing Age of
  Information in Broadcast Wireless Networks}}.
\newblock \bibinfo{journal}{\emph{IEEE/ACM Transactions on Networking}}
  \bibinfo{volume}{26}, \bibinfo{number}{6} (\bibinfo{year}{2018}),
  \bibinfo{pages}{2637--2650}.
\newblock
\urldef\tempurl%
\url{https://doi.org/10.1109/TNET.2018.2873606}
\showDOI{\tempurl}


\bibitem[Kalør and Popovski(2019)]%
        {popovski2019}
\bibfield{author}{\bibinfo{person}{Anders~E. Kalør} {and}
  \bibinfo{person}{Petar Popovski}.} \bibinfo{year}{2019}\natexlab{}.
\newblock \showarticletitle{{Minimizing the Age of Information From Sensors
  With Common Observations}}.
\newblock \bibinfo{journal}{\emph{IEEE Wireless Communications Letters}}
  \bibinfo{volume}{8}, \bibinfo{number}{5} (\bibinfo{year}{2019}),
  \bibinfo{pages}{1390--1393}.
\newblock
\urldef\tempurl%
\url{https://doi.org/10.1109/LWC.2019.2919528}
\showDOI{\tempurl}


\bibitem[Kaul et~al\mbox{.}(2012)]%
        {yates2012}
\bibfield{author}{\bibinfo{person}{Sanjit Kaul}, \bibinfo{person}{Roy Yates},
  {and} \bibinfo{person}{Marco Gruteser}.} \bibinfo{year}{2012}\natexlab{}.
\newblock \showarticletitle{{Real-time status: How often should one update?}}.
  In \bibinfo{booktitle}{\emph{2012 Proceedings IEEE INFOCOM}}.
  \bibinfo{pages}{2731--2735}.
\newblock
\urldef\tempurl%
\url{https://doi.org/10.1109/INFCOM.2012.6195689}
\showDOI{\tempurl}


\bibitem[Kingma and Ba(2017)]%
        {adam}
\bibfield{author}{\bibinfo{person}{Diederik~P. Kingma} {and}
  \bibinfo{person}{Jimmy Ba}.} \bibinfo{year}{2017}\natexlab{}.
\newblock \bibinfo{title}{{Adam: A Method for Stochastic Optimization}}.
\newblock
\newblock
\showeprint[arxiv]{1412.6980}~[cs.LG]


\bibitem[Kraft(1988)]%
        {slsqp}
\bibfield{author}{\bibinfo{person}{D. Kraft}.} \bibinfo{year}{1988}\natexlab{}.
\newblock \bibinfo{booktitle}{\emph{A Software Package for Sequential Quadratic
  Programming}}.
\newblock \bibinfo{publisher}{Wiss. Berichtswesen d. DFVLR}.
\newblock
\urldef\tempurl%
\url{https://books.google.com/books?id=4rKaGwAACAAJ}
\showURL{%
\tempurl}


\bibitem[Li et~al\mbox{.}(2022)]%
        {sch-li}
\bibfield{author}{\bibinfo{person}{Chengzhang Li}, \bibinfo{person}{Qingyu
  Liu}, \bibinfo{person}{Shaoran Li}, \bibinfo{person}{Yongce Chen},
  \bibinfo{person}{Y.~Thomas Hou}, \bibinfo{person}{Wenjing Lou}, {and}
  \bibinfo{person}{Sastry Kompella}.} \bibinfo{year}{2022}\natexlab{}.
\newblock \showarticletitle{{Scheduling With Age of Information Guarantee}}.
\newblock \bibinfo{journal}{\emph{IEEE/ACM Transactions on Networking}}
  \bibinfo{volume}{30}, \bibinfo{number}{5} (\bibinfo{year}{2022}),
  \bibinfo{pages}{2046--2059}.
\newblock
\urldef\tempurl%
\url{https://doi.org/10.1109/TNET.2022.3156866}
\showDOI{\tempurl}


\bibitem[Maatouk et~al\mbox{.}(2019)]%
        {8845254}
\bibfield{author}{\bibinfo{person}{Ali Maatouk}, \bibinfo{person}{Mohamad
  Assaad}, {and} \bibinfo{person}{Anthony Ephremides}.}
  \bibinfo{year}{2019}\natexlab{}.
\newblock \showarticletitle{Minimizing The Age of Information: NOMA or OMA?}.
  In \bibinfo{booktitle}{\emph{IEEE INFOCOM 2019 - IEEE Conference on Computer
  Communications Workshops (INFOCOM WKSHPS)}}. \bibinfo{pages}{102--108}.
\newblock
\urldef\tempurl%
\url{https://doi.org/10.1109/INFCOMW.2019.8845254}
\showDOI{\tempurl}


\bibitem[Maatouk et~al\mbox{.}(2020a)]%
        {9007478}
\bibfield{author}{\bibinfo{person}{Ali Maatouk}, \bibinfo{person}{Mohamad
  Assaad}, {and} \bibinfo{person}{Anthony Ephremides}.}
  \bibinfo{year}{2020}\natexlab{a}.
\newblock \showarticletitle{On the Age of Information in a CSMA Environment}.
\newblock \bibinfo{journal}{\emph{IEEE/ACM Transactions on Networking}}
  \bibinfo{volume}{28}, \bibinfo{number}{2} (\bibinfo{year}{2020}),
  \bibinfo{pages}{818--831}.
\newblock
\urldef\tempurl%
\url{https://doi.org/10.1109/TNET.2020.2971350}
\showDOI{\tempurl}


\bibitem[Maatouk et~al\mbox{.}(2020b)]%
        {9137714}
\bibfield{author}{\bibinfo{person}{Ali Maatouk}, \bibinfo{person}{Saad
  Kriouile}, \bibinfo{person}{Mohamad Assaad}, {and} \bibinfo{person}{Anthony
  Ephremides}.} \bibinfo{year}{2020}\natexlab{b}.
\newblock \showarticletitle{The Age of Incorrect Information: A New Performance
  Metric for Status Updates}.
\newblock \bibinfo{journal}{\emph{IEEE/ACM Transactions on Networking}}
  \bibinfo{volume}{28}, \bibinfo{number}{5} (\bibinfo{year}{2020}),
  \bibinfo{pages}{2215--2228}.
\newblock
\urldef\tempurl%
\url{https://doi.org/10.1109/TNET.2020.3005549}
\showDOI{\tempurl}


\bibitem[Najm et~al\mbox{.}(2018)]%
        {najm2018}
\bibfield{author}{\bibinfo{person}{Elie Najm}, \bibinfo{person}{Rajai Nasser},
  {and} \bibinfo{person}{Emre Telatar}.} \bibinfo{year}{2018}\natexlab{}.
\newblock \showarticletitle{{Content Based Status Updates}}. In
  \bibinfo{booktitle}{\emph{2018 IEEE International Symposium on Information
  Theory (ISIT)}}. \bibinfo{pages}{2266--2270}.
\newblock
\urldef\tempurl%
\url{https://doi.org/10.1109/ISIT.2018.8437577}
\showDOI{\tempurl}


\bibitem[Pan et~al\mbox{.}(2021)]%
        {sch-sun}
\bibfield{author}{\bibinfo{person}{Jiayu Pan}, \bibinfo{person}{Ahmed~M.
  Bedewy}, \bibinfo{person}{Yin Sun}, {and} \bibinfo{person}{Ness~B. Shroff}.}
  \bibinfo{year}{2021}\natexlab{}.
\newblock \showarticletitle{Minimizing Age of Information via Scheduling over
  Heterogeneous Channels}. In \bibinfo{booktitle}{\emph{Proceedings of the
  Twenty-Second International Symposium on Theory, Algorithmic Foundations, and
  Protocol Design for Mobile Networks and Mobile Computing}} (Shanghai, China)
  \emph{(\bibinfo{series}{MobiHoc '21})}. \bibinfo{publisher}{Association for
  Computing Machinery}, \bibinfo{address}{New York, NY, USA},
  \bibinfo{pages}{111–120}.
\newblock
\showISBNx{9781450385589}
\urldef\tempurl%
\url{https://doi.org/10.1145/3466772.3467040}
\showDOI{\tempurl}


\bibitem[Rafiee et~al\mbox{.}(2024)]%
        {const-parisa}
\bibfield{author}{\bibinfo{person}{Parisa Rafiee}, \bibinfo{person}{Zhuoxuan
  Ju}, {and} \bibinfo{person}{Miloš Doroslovački}.}
  \bibinfo{year}{2024}\natexlab{}.
\newblock \showarticletitle{Adaptive ON/OFF Scheduling to Minimize Age of
  Information in an Energy-Harvesting Receiver}.
\newblock \bibinfo{journal}{\emph{IEEE Sensors Journal}} \bibinfo{volume}{24},
  \bibinfo{number}{3} (\bibinfo{year}{2024}), \bibinfo{pages}{3898--3911}.
\newblock
\urldef\tempurl%
\url{https://doi.org/10.1109/JSEN.2023.3339598}
\showDOI{\tempurl}


\bibitem[Raghavendra et~al\mbox{.}(2004)]%
        {wirelessnetworks}
\bibfield{editor}{\bibinfo{person}{C.~S. Raghavendra},
  \bibinfo{person}{Krishna~M. Sivalingam}, {and} \bibinfo{person}{Taieb Znati}}
  (Eds.). \bibinfo{year}{2004}\natexlab{}.
\newblock \bibinfo{booktitle}{\emph{{Wireless sensor networks}}}.
\newblock \bibinfo{publisher}{Kluwer Academic Publishers},
  \bibinfo{address}{USA}.
\newblock
\showISBNx{1402078838}


\bibitem[Ramakanth et~al\mbox{.}(2023)]%
        {ramakanth2023monitoring}
\bibfield{author}{\bibinfo{person}{R~Vallabh Ramakanth},
  \bibinfo{person}{Vishrant Tripathi}, {and} \bibinfo{person}{Eytan Modiano}.}
  \bibinfo{year}{2023}\natexlab{}.
\newblock \bibinfo{title}{Monitoring Correlated Sources: AoI-based Scheduling
  is Nearly Optimal}.
\newblock
\newblock
\showeprint[arxiv]{2312.16813}~[cs.NI]


\bibitem[Soysal and Ulukus(2021)]%
        {soysal2019}
\bibfield{author}{\bibinfo{person}{Alkan Soysal} {and} \bibinfo{person}{Sennur
  Ulukus}.} \bibinfo{year}{2021}\natexlab{}.
\newblock \showarticletitle{{Age of Information in G/G/1/1 Systems: Age
  Expressions, Bounds, Special Cases, and Optimization}}.
\newblock \bibinfo{journal}{\emph{IEEE Trans. Inf. Theor.}}
  \bibinfo{volume}{67}, \bibinfo{number}{11} (\bibinfo{date}{nov}
  \bibinfo{year}{2021}), \bibinfo{pages}{7477–7489}.
\newblock
\showISSN{0018-9448}
\urldef\tempurl%
\url{https://doi.org/10.1109/TIT.2021.3095823}
\showDOI{\tempurl}


\bibitem[Sun et~al\mbox{.}(2017)]%
        {sun-error}
\bibfield{author}{\bibinfo{person}{Yin Sun}, \bibinfo{person}{Yury Polyanskiy},
  {and} \bibinfo{person}{Elif Uysal-Biyikoglu}.}
  \bibinfo{year}{2017}\natexlab{}.
\newblock \showarticletitle{Remote estimation of the Wiener process over a
  channel with random delay}. In \bibinfo{booktitle}{\emph{2017 IEEE
  International Symposium on Information Theory (ISIT)}}.
  \bibinfo{pages}{321--325}.
\newblock
\urldef\tempurl%
\url{https://doi.org/10.1109/ISIT.2017.8006542}
\showDOI{\tempurl}


\bibitem[Sun et~al\mbox{.}(2016)]%
        {sun2016}
\bibfield{author}{\bibinfo{person}{Yin Sun}, \bibinfo{person}{Elif
  Uysal-Biyikoglu}, \bibinfo{person}{Roy Yates}, \bibinfo{person}{C.~Emre
  Koksal}, {and} \bibinfo{person}{Ness~B. Shroff}.}
  \bibinfo{year}{2016}\natexlab{}.
\newblock \showarticletitle{{Update or wait: How to keep your data fresh}}. In
  \bibinfo{booktitle}{\emph{IEEE INFOCOM 2016 - The 35th Annual IEEE
  International Conference on Computer Communications}}. \bibinfo{pages}{1--9}.
\newblock
\urldef\tempurl%
\url{https://doi.org/10.1109/INFOCOM.2016.7524524}
\showDOI{\tempurl}


\bibitem[Taylor and Karlin(2011)]%
        {embeddedmarkov}
\bibfield{author}{\bibinfo{person}{Howard~M. Taylor} {and}
  \bibinfo{person}{Samuel Karlin}.} \bibinfo{year}{2011}\natexlab{}.
\newblock \bibinfo{booktitle}{\emph{An Introduction To Stochastic Modeling}
  (\bibinfo{edition}{fourth edition} ed.)}.
\newblock \bibinfo{publisher}{Academic Press}.
\newblock


\bibitem[Tong et~al\mbox{.}(2022)]%
        {tong2022}
\bibfield{author}{\bibinfo{person}{Jingwen Tong}, \bibinfo{person}{Liqun Fu},
  {and} \bibinfo{person}{Zhu Han}.} \bibinfo{year}{2022}\natexlab{}.
\newblock \showarticletitle{{Age-of-Information Oriented Scheduling for
  Multichannel IoT Systems With Correlated Sources}}.
\newblock \bibinfo{journal}{\emph{IEEE Transactions on Wireless
  Communications}} \bibinfo{volume}{21}, \bibinfo{number}{11}
  (\bibinfo{year}{2022}), \bibinfo{pages}{9775--9790}.
\newblock
\urldef\tempurl%
\url{https://doi.org/10.1109/TWC.2022.3179305}
\showDOI{\tempurl}


\bibitem[Tripathi and Modiano(2022)]%
        {modiano2022}
\bibfield{author}{\bibinfo{person}{Vishrant Tripathi} {and}
  \bibinfo{person}{Eytan Modiano}.} \bibinfo{year}{2022}\natexlab{}.
\newblock \showarticletitle{{Optimizing age of information with correlated
  sources}}. In \bibinfo{booktitle}{\emph{Proceedings of the Twenty-Third
  International Symposium on Theory, Algorithmic Foundations, and Protocol
  Design for Mobile Networks and Mobile Computing}} (Seoul, Republic of Korea)
  \emph{(\bibinfo{series}{MobiHoc '22})}. \bibinfo{publisher}{Association for
  Computing Machinery}, \bibinfo{address}{New York, NY, USA},
  \bibinfo{pages}{41–50}.
\newblock
\showISBNx{9781450391658}
\urldef\tempurl%
\url{https://doi.org/10.1145/3492866.3549719}
\showDOI{\tempurl}


\bibitem[Yates and Kaul(2019)]%
        {yates2019}
\bibfield{author}{\bibinfo{person}{Roy~D. Yates} {and}
  \bibinfo{person}{Sanjit~K. Kaul}.} \bibinfo{year}{2019}\natexlab{}.
\newblock \showarticletitle{{The Age of Information: Real-Time Status Updating
  by Multiple Sources}}.
\newblock \bibinfo{journal}{\emph{IEEE Transactions on Information Theory}}
  \bibinfo{volume}{65}, \bibinfo{number}{3} (\bibinfo{year}{2019}),
  \bibinfo{pages}{1807--1827}.
\newblock
\urldef\tempurl%
\url{https://doi.org/10.1109/TIT.2018.2871079}
\showDOI{\tempurl}


\bibitem[Zou et~al\mbox{.}(2023)]%
        {zou2023costly}
\bibfield{author}{\bibinfo{person}{Peng Zou}, \bibinfo{person}{Ali Maatouk},
  \bibinfo{person}{Jin Zhang}, {and} \bibinfo{person}{Suresh Subramaniam}.}
  \bibinfo{year}{2023}\natexlab{}.
\newblock \showarticletitle{{How Costly Was That (In)Decision?}}. In
  \bibinfo{booktitle}{\emph{2023 21st International Symposium on Modeling and
  Optimization in Mobile, Ad Hoc, and Wireless Networks (WiOpt)}}.
  \bibinfo{pages}{278--285}.
\newblock
\urldef\tempurl%
\url{https://doi.org/10.23919/WiOpt58741.2023.10349833}
\showDOI{\tempurl}


\end{thebibliography}

\appendix
\section{Simulation Results for Impact of Buffer Length on Age of Information} 
To gain insights into the impact of buffer length on system performance, we conduct experiments using various buffer configurations, ranging from zero buffer to larger buffer lengths. We compare the performance of a zero buffer to that of a single buffer below, but similar results were observed for other buffer sizes. Our primary focus is on evaluating the effect of buffer length on AoI and our numerical simulations suggest that having a buffer is not always optimal in our system. To illustrate this, we compared zero-buffer and one-buffer systems with $2$ processes, $2$ correlated sensors and a fixed service rate $\mu= 2.5$.
We also set $\mathbf{P_C}$ as follows: 
\[
\boldsymbol{\mathbf{P_C}} = \begin{bmatrix}
1 & 0.5 \\
0.5 & 1
\end{bmatrix},
\]
The 3D plot in Figure \ref{fig:zerobuffer} illustrates that the zero-buffer configuration generally outperforms the one-buffer setup in minimizing total AoI across varying arrival rates, with the zero-buffer mostly resulting in lower AoI values. This suggests that, despite the additional queuing capacity provided by a buffer, the associated delay can cancel its potential benefits, particularly as arrival rates increase. The findings underscore that while a buffer may be beneficial in specific scenarios, its optimality is highly dependent on system conditions, and in many cases, the zero-buffer approach remains superior for minimizing AoI. Given all the above, we adopted the no-buffer assumption in our paper. The optimal queuing scheme for multi-correlated sources within the AoI literature remains an open question, which we will address in our future work.

\begin{figure}[h]
  \centering
  \includegraphics[width=0.35\textwidth]{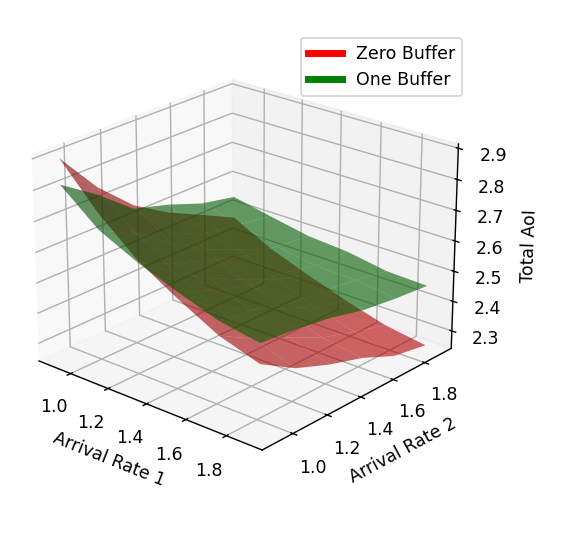}
  \caption{Total AoI in the 2 sensors and 2 processes system, $\mu= 2.5$.}
  \label{fig:zerobuffer}
\end{figure}

\section{Proof of Lemma \ref{Lem1} } 
\label{reduction-P}
To prove our lemma, we apply the splitting property of the Poisson process. Let $N(t)$ be a Poisson process with rate parameter $\lambda$. If events are split into two groups with probabilities $p$ and $1-p$, then the resulting processes $N_1(t)$ and $N_2(t)$ are independent Poisson processes with rate parameters $p\lambda$ and $(1-p)\lambda$ respectively\cite{splitting_poisson}.

From process $j$ perspective, we can split arrivals from sensor $i$ into two groups as informative and uninformative arrivals with probabilities $p^c_{ij}$ and $1-p^c_{ij}$. The rate of arrivals from sensor $i$ is $\lambda_i$, so the rate of informative arrivals for process $j$ from sensor $i$ is $p^c_{ij}\lambda_j$. Then, we can merge all informative arrivals as a single process because all of them are Poisson processes. Total arrival rate of packets with information of process $j$ is denoted by $\lambda_{j}^*$ and it equals to

\begin{equation}
\lambda_{j}^* = \sum_{i=1}^{N} p^c_{ij}\lambda_i
\end{equation}
As a vector,
\begin{equation}
{\boldsymbol{\lambda^*}}^T = \begin{bmatrix}
\lambda_{1}^* & \lambda_{2}^* & \dots & \lambda_{M}^*\end{bmatrix} = \boldsymbol{\lambda}^T\mathbf{P_C}
\end{equation}

The importance of the packet is whether it has information of process $j$ so  we can say that The system with $N$ sensors and arrival rates $\boldsymbol{\lambda}$ shown in Figure \ref{fig:systemmodel} equivalents to the system with two sources and arrival rates of $\lambda^*_j$ and $\lambda_C - \lambda^*_j$ shown in Figure \ref{fig:equv_model} from process $j$'s perspective
\section{Proof of Theorem \ref{The1} } \label{aoi-theorem}

We need to find $\mathbb{E}[\Tilde{Y}_j]$ and $\mathbb{E}[\Tilde{Y}_j^2]$. First, let $\mathrm{A}$ be a set defined as a set that consists of indices corresponding to arrivals between the $(l-1)$-th and $l$-th informative arrivals. Note that $l$-th informative arrival may not correspond to $l$-th arrival if there is any uninformative arrival before it. So, let the $l$-th informative arrival correspond to the $k$-th  arrival, and let the $(l-1)$-th informative arrival correspond to the $(k-r)$-th  arrival.
$A$ can be rewritten as,
\begin{equation}
A = \{a \mid k-r < a \leq k \text{ for } r \geq 1 \text { and } k \geq r\} 
\end{equation}
where $P( r) = (1-\Tilde{p_j})^{r-1}\Tilde{p_j} \quad \text{for }  r \geq 1$. From \cite{mm1}, we also know that  
\begin{align}
    \mathbb{E}[Y] = \frac{1}{\lambda_C}+\frac{1}{\mu}
    \label{expectedy}
    \end{align}
\begin{align}
\mathbb{E}[Y^2] = \frac{2(\lambda_C^2+\lambda_C\mu+\mu^2)}{\lambda_C^2\mu^2} 
    \label{expectedy2}
\end{align}

When we express $\Tilde{Y}_{j}^l$ in terms of $Y_{a}$, we get

\begin{equation}
\Tilde{Y}_{j}^l = \sum_{a \in A} Y_a 
\end{equation}
\begin{align}
\nonumber \mathbb{E}[\Tilde{Y}_j] = \sum_{r=1}^{\infty}P(R = r)r\mathbb{E}[Y] \\\nonumber
= \mathbb{E}[Y]\sum_{r=1}^{\infty}(1-\Tilde{p_j})^{r-1}\Tilde{p_j}r\\\nonumber
= \mathbb{E}[Y]\mathbb{E}[R]\\
= \mathbb{E}[Y]/\Tilde{p_j} \end{align}

Similarly, 
\begin{align}
\mathbb{E}[\Tilde{Y}_i^2|r] = \mathbb{E}[(\sum_{a \in A}Y_a)^2] = r\mathbb{E}[Y^2]+ 2\binom{r}{2}\mathbb{E}[Y]^2
\end{align}

\begin{align}
\mathbb{E}[\Tilde{Y}_i^2] = \sum_{r=1}^{\infty}\Tilde{p_i}(1-\Tilde{p_j})^r\mathbb{E}[\Tilde{Y}_i^2|R=r] \nonumber\\
= \frac{\mathbb{E}[Y^2]}{\Tilde{p_j}} + \frac{2\mathbb{E}[Y]^2(1-\Tilde{p_j})}{\Tilde{p_j}^2}
\end{align}

After that, we replace $\mathbb{E}[\Tilde{Y}_j]$ and $\mathbb{E}[\Tilde{Y}_j^2]$ to findings in equation (\ref{ageofinformation}) to get
  \begin{align}
 \Delta_i = \frac{\lambda_C}{\lambda_C+\mu}\big(\frac{\mu\mathbb{E}[Y^2]}{2} + \frac{\mu\mathbb{E}[Y]^2(1-\Tilde{p_j})}{\Tilde{p_j}} + \mathbb{E}[Y]\big)
\end{align}

\section{Proof of Lemma \ref{Lem3} } \label{statechange-P}

Both state change and service time events are exponential random variables, so they are memoryless. Using the memoryless property, the probability distribution function of the remaining service time $t$ is

\begin{align}
    f_T(t | \mu) = \mu e^{-t\mu}
\end{align}

The number of state changes $k$ over time $t$ is a Poisson process is

\begin{align}
P(k  | t) = \frac{e^{-(\zeta t)} (\zeta t)^k}{k!}
\end{align}

Next, we determine the probabilities of transitioning from the initial state $i$ to the final state $j$ after $k$ transitions by evaluating $(\boldsymbol{\Omega}^k)_{ij}$. After that, we get the matrix consisting of transitioning probabilities from state $i$ to state $j$ when service time is $t$ and the number of  state changes $k$ as
\begin{align}
\mathbf{P_N}(k,t) =  \mu e^{-\mu t}  \frac{e^{-(t\zeta)} (t\zeta)^k}{k!}\boldsymbol{\Omega}^k dt
\end{align}
\allowdisplaybreaks

Then, we obtain
\begin{align}
\mathbf{P_N} = \int_{0}^{\infty} \mu e^{-\mu t} \sum_{k=0}^{\infty} \frac{e^{-(t\zeta)} (t\zeta)^k}{k!}\boldsymbol{\Omega}^k dt \nonumber\\\nonumber
=\sum_{k=0}^{\infty} \int_{0}^{\infty} \mu e^{-\mu t}  \frac{e^{-(t\zeta)} (t\zeta)^k}{k!}\boldsymbol{\Omega}^k dt \\\nonumber
= \sum_{k=0}^{\infty} \frac{\boldsymbol{\Omega}^k}{k!} \int_{0}^{\infty} \mu e^{-\mu t}e^{-(t\zeta)} (t\zeta)^k dt \\\nonumber
= \sum_{k=0}^{\infty} \frac{\boldsymbol{\Omega}^k}{k!} \operatorname{\Gamma}\left(k+1,0\right){\zeta}^k{\mu}\cdot\left({\mu}+{\zeta}\right)^{-k-1} \\\nonumber \text{where $\operatorname{\Gamma}\left(k+1,0\right) = k!$}\\\nonumber
= \sum_{k=0}^{\infty} \boldsymbol{\Omega}^k{\zeta}^k{\mu}\cdot\left({\mu}+{\zeta}\right)^{-k-1}\\\nonumber
= \frac{\mu}{\mu+\zeta} \sum_{k=0}^{\infty} \left(\frac{\zeta\boldsymbol{\Omega}}{\mu+\zeta}\right)^k\\\nonumber
= \frac{\mu}{\mu+\zeta}  \sum_{k=0}^{\infty}\left(\frac{\zeta\boldsymbol{\Omega}}{\mu+\zeta}\right)^k \\
= \frac{\mu}{\mu+\zeta} \left(\mathbf{I} - \frac{\zeta\boldsymbol{\Omega}}{\mu+\zeta}\right)^{-1} 
\end{align}

\section{Proof of Lemma \ref{Lem4}} \label{transition_matrices}
 
\begin{figure*}[h]
    \centering
    \begin{subfigure}{0.35\textwidth}
        \centering
        \includegraphics[width=\linewidth]{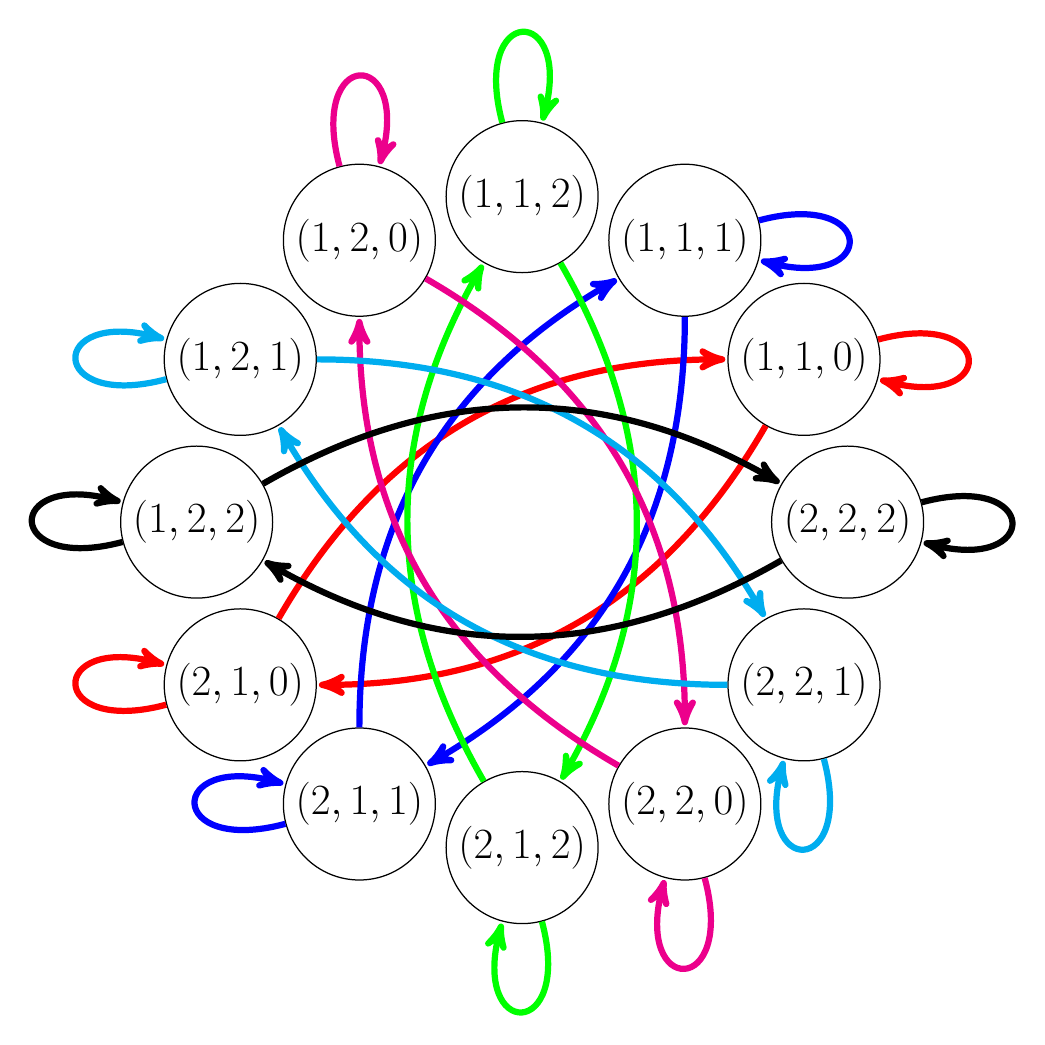}
        \caption{Possible transitions when a state change occurs.}
        \label{markov_state_change}
    \end{subfigure}%
    \begin{subfigure}{0.35\textwidth}
        \centering
        \includegraphics[width=\linewidth]{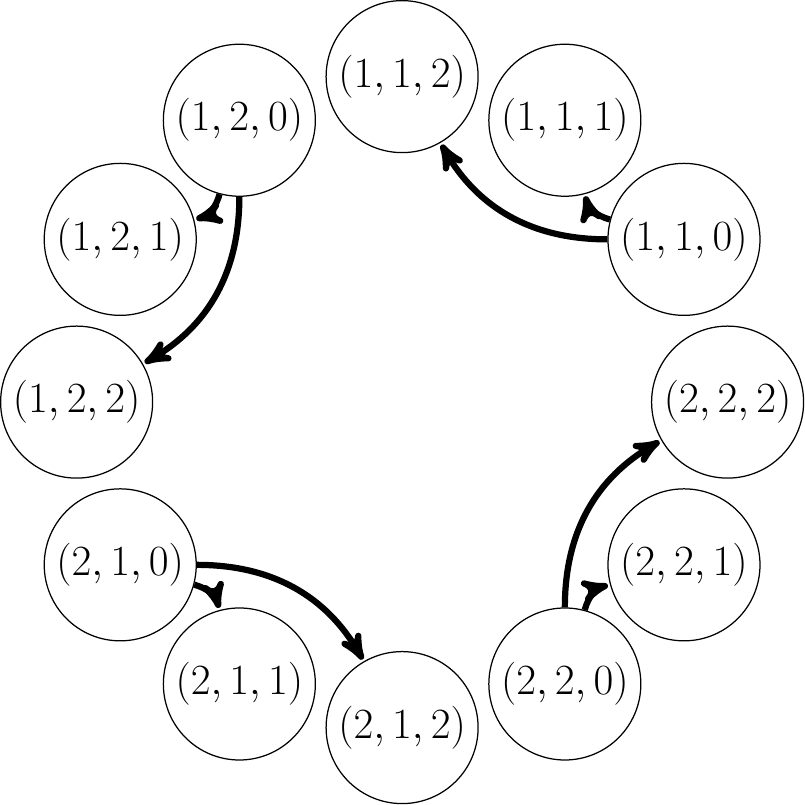}
        \caption{Possible transitions when a packet arrives.}
        \label{markov_arrivals}
    \end{subfigure}%
    
    \begin{subfigure}{0.35\textwidth}
        \centering
        \includegraphics[width=\linewidth]{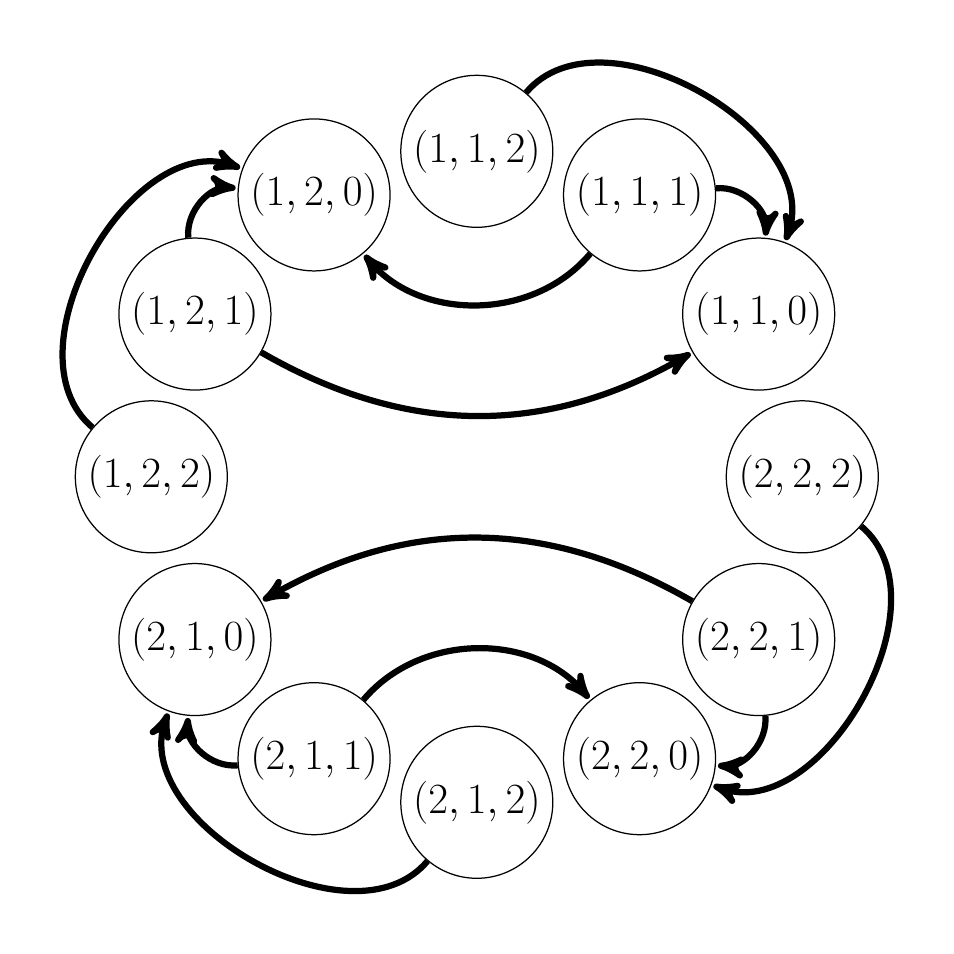}
        \caption{Possible transitions when a packet departs.}
        \label{markov_departures}
    \end{subfigure}%
    \begin{subfigure}{0.35\textwidth}
        \centering
        \includegraphics[width=\linewidth]{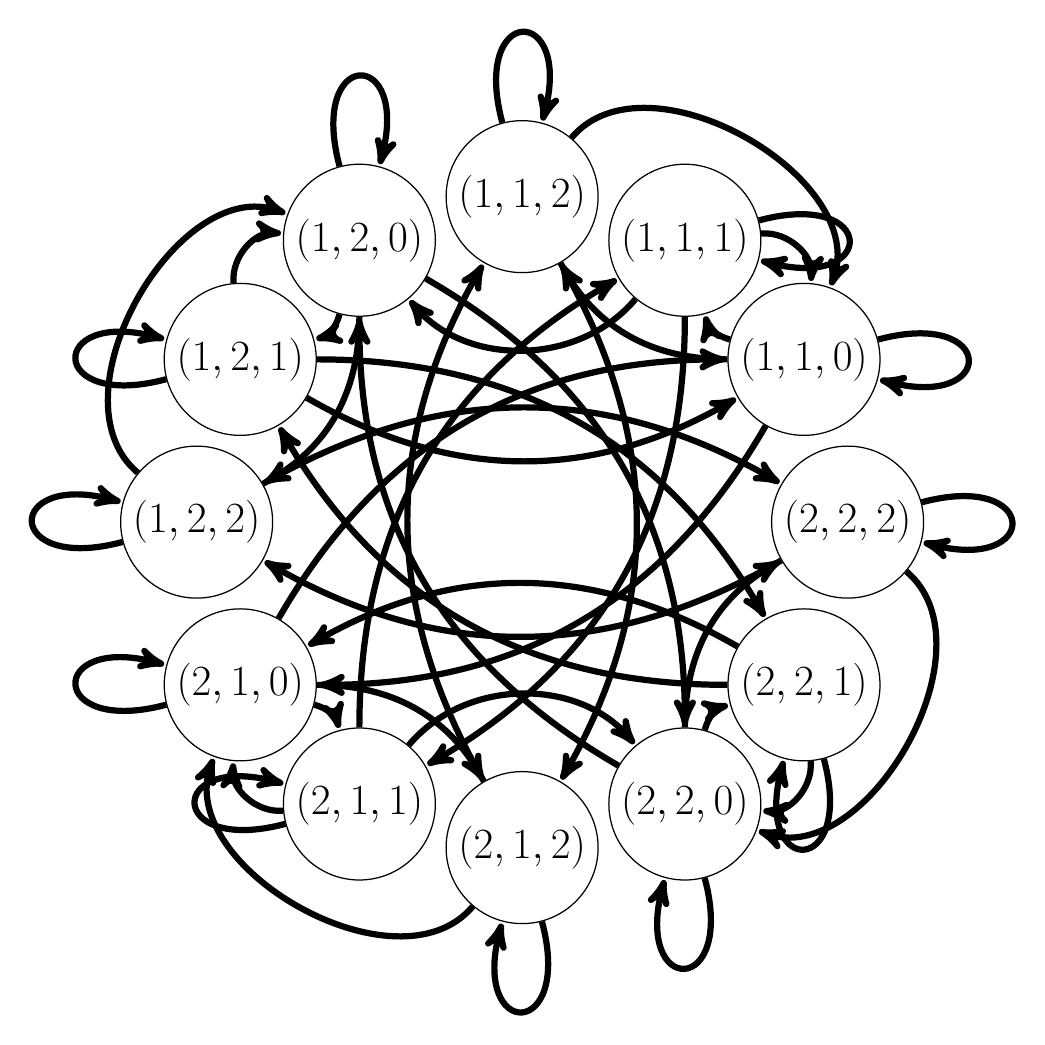}
        \caption{All transitions together.}
        \label{markov_all}
    \end{subfigure}%
    \caption{Transitions of $\mathbf{P_M}$.}
    \label{pm_chains}
\end{figure*}

To understand transitions and their probabilities, first, we look at events that may occur and cause a state transition. There are three events in the system: packet arrival, packet departure, and state change for the given process. If the server is idle and an event occurs, the event may be packet arrival or state change for the given process with probabilities $\frac{\lambda_{C}}{\zeta + \lambda_C}$ or $\frac{\zeta}{\zeta + \lambda_C}$, respectively but it can not be packet departure. The arrivals may be informative and change the information the monitor has or not with probabilities $\Tilde{p_i} = \frac{\lambda_{i}^*}{\lambda_C}$ or $(1 -\Tilde{p_i}) = \frac{\lambda_C-\lambda_{i}^*}{\lambda_C}$, respectively.  Conversely, if the server is busy with an informative or uninformative packet and an event occurs, the event may be packet departure or state change for the given process with probabilities $\frac{\mu}{\zeta + \mu}$ or $\frac{\zeta}{\zeta + \mu}$, respectively but it can not be packet arrival. Then, we obtain 7 different possibilities that cause transitions. They are:
\begin{enumerate}
    \item \textbf{Server is idle, and state change for the given process happens. }
    \item \textbf{Server is busy with an informative packet, and state change for the given process happens.}
    \item  \textbf{Server is uninformative with an informative packet, and state change for the given process happens.}
    \item  \textbf{Server is idle, and an informative packet arrives.}
    \item  \textbf{Server is idle, and an uninformative packet arrives.}
    \item \textbf{Server is busy with an informative packet, and the packet departs.}
    \item  \textbf{Server is uninformative with an informative packet, and the packet departs.}
\end{enumerate}

To understand transitions better, we illustrate a Markov chain for states $(x,y,z)$ when $K=2$ in Figure \ref{pm_chains}. The transitions caused by the process' state change are shown in Figure (\ref{markov_state_change}). Change in the process' state may affect only the current state $x$ value and the state information the monitor has $y$, server state $z$ stays constant when the process' state changes. In Figure (\ref{markov_arrivals}), the transitions caused by packet arrivals are shown. Packet arrival has an effect on only server state $z$, and $z$ goes from $0$ to $1$ if the packet is informative. Otherwise, $z$ goes from $0$ to $2$. The transitions caused by packet departures are shown in Figure (\ref{markov_departures}). If the packet is uninformative, $z$ goes to $0$ from $2$ because the server becomes idle and $x,y$ stays constant. However, if the packet is informative, the packet has the information when it arrives and it updates $y$ while departing. So, $z$ goes to $0$ from $1$, and $y$ goes to the value the packet has, which can be all possible values with some probability. Lastly, we show all transitions in Figure (\ref{markov_all}).

Then we calculate each transition probability to form $\mathbf{P_M}$ matrix. First, state change does not affect the information the monitor has and the server's state, which are $y$ and $z$, but it causes a transition for $x$ from initial $x_1$ to $x_2$ with probability $\Omega_{x_1x_2}$ for every \( x_1,x_2 = 1, \ldots, K \). After having those probabilities, we get equations (\ref{0_state_eq}),  (\ref{1_state_eq}), and (\ref{2_state_eq}) as multiplication of $\Omega_{x_1x_2}$ and probability of having state change before packet arrival or departure. 

Packet arrival does not affect the current state of the process and the information the monitor has, which are $x$ and $y$, but it causes a transition for $z$ from initial $0$ to $1$ or $2$ with probabilities $\frac{\lambda_{i}^*}{\lambda_C}$ and $\frac{\lambda_C-\lambda_{i}^*}{\lambda_C}$, respectively. Then, we get the equations (\ref{informative_arrival_eq}) and (\ref{uninformative_arrival_eq}), as multiplication of the probability of being informative or uninformative and probability of having packet arrival before state change. 

Uninformative packet departure does not affect the current state of the process $x$ and the information the monitor has $y$, but it causes a transition for $z$ from initial $2$ to $0$. Therefore, we get equation (\ref{uninformative_departure_eq}) as the probability of having packet departure before state change.

Informative packet departure does not affect the current state of the process $x$, but it causes a transition for $(y,z)$ from initial $(y_1,1)$ to $(y_2,0)$ with some probability. To find this, we define two events, $X_1$ and $X_2$, as the state of the process when a packet arrives and the state of the process when the packet departures. Then, the probability of transition from $(x_1,y_1,1)$ to $(x_1,y_2,0)$ if packet departure happens before state change of process is $P[X_1 = y_2 | X_2 = x_1]$ because the packet has the information of state when it arrives. We calculated $P[X_2 = x_2 | X_1 = y_2] = p^n_{y_2x_1}$ in Lemma \ref{Lem3} which is the probability of the process' state is $x_2$ after the packet is served when the process' state is $y_2$ initially. Using Bayes' rule, we obtain:

\begin{equation}
P[X_1 = y_2 | X_2 = x_1] = \frac{P[X_2 = x_1 | X_1 = y_2] \cdot P[X_1 = y_2]}{P[X_2 = x_1]}
\end{equation}

where $P[X_2 = x_1] = \psi_{x_1}$ and $P[X_1 = y_2] = \psi_{y_2}$ in equation (\ref{stationary_state}). After getting $P[X_1 = y_2 | X_2 = x_1]$ as $p^n_{y_2x_1}\frac{\psi_{y_2}}{\psi_{x_1}}$, we multiply it by the probability of packet departure before state change of process to get equation (\ref{informative_departure_eq}).

\section{Proof of Lemma \ref{Lem5}} \label{expected_times}
Time is spent in a state at every visit until an event happens. All three events have exponential distribution between two occurrences, and they are Poisson processes. Let the event of the process' state change, the event of arrival, and the event of departure be $N_S(t)$ with rate parameter $\zeta$, $N_A(t)$ with rate parameter $\lambda_C$, $N_D(t)$ with rate parameter $\mu$, respectively. Using the merging property of the Poisson process, we obtain the event occurrence process for each state. The event of the process' state change can happen in all states. The event of arrival happens only when the server is idle, which is z = 0. The event of departure happens only when the server is busy, which is z = 1 and z = 2. 

For any x and y when z = 0, the possible events are the event of arrival or the event of the process' state change, so the event occurrence process becomes $N_S(t)$+$N_A(t)$, which is a Poisson process with rate $\zeta + \lambda_C$ and the expected interarrival time is equal to $\mathbb{E}[T_{(x, y, 0)}]$. Thus,
For every \( x, y = 1,\ldots, K \):

\begin{equation}
\mathbb{E}[T_{(x, y, 0)}] = \frac{1}{\zeta+\lambda_C} 
\end{equation}

For any x and y when z = 1 or z = 2, the possible events are the event of departure or the event of the process' state change, so the event occurrence process becomes $N_D(t)$+$N_A(t)$, which is a Poisson process with rate $\mu + \lambda_C$ and the expected interarrival time is equal to $\mathbb{E}[T_{(x, y, 1)}]$ or $\mathbb{E}[T_{(x, y, 2)}]$. Thus,
For every \( x, y = 1,\ldots, K \):

\begin{equation}
\mathbb{E}[T_{(x, y, 0)}] = \frac{1}{\zeta+\lambda_C} 
\end{equation}

\begin{equation}
\mathbb{E}[T_{(x, y, 1)}] = \mathbb{E}[T_{(x, y, 2)}]  = \frac{1}{\zeta+\mu}  
\end{equation}

\section{Proof of Lemma \ref{Lem6}} \label{convexity-P}

As a first step, we aim to prove that the function \(g(\mathbf{\Tilde{p}})\) defined as \(g(\mathbf{\Tilde{p}}) = \frac{1}{\Tilde{p}_1} + \frac{1}{\Tilde{p}_2} + \ldots + \frac{1}{\Tilde{p}_n}\) is complex using its Hessian matrix.The Hessian matrix of \(g\) is given by:
\[
H_g = 
\begin{bmatrix}
\frac{2}{\Tilde{p}_1^3} & 0& \cdots & 0 \\
0 & \frac{2}{\Tilde{p}_2^3} & \cdots & 0 \\
\vdots & \vdots & \ddots & \vdots \\ 0 & 0 & \cdots & \frac{2}{\Tilde{p}_N^3}
\end{bmatrix}
\]

All the diagonal elements of the Hessian matrix are positive when $\Tilde{p}_i$ is positive for all $i$. We know that $\Tilde{p}_i$ is formed using probabilities and positive arrival rates so all the diagonal elements of the matrix are positive. In addition,  all the off-diagonal elements are zero. This makes the Hessian matrix positive-definite and the function $g$ convex. After that, we want to show that $f(\mathbf{P_C}) = g(\frac{\mathbf{P_C}\boldsymbol{\lambda}}{\lambda_C})$ is convex. $g(\frac{\mathbf{P_C}\boldsymbol{\lambda}}{\lambda_C})$ is an affine composition of $g$ so the operation preserves the functions convexity.

\end{document}